\documentclass[journal]{IEEEtran}

\usepackage{amsmath}
\usepackage{amsfonts}
\usepackage{algorithm}
\usepackage{algorithmic}
\usepackage{cite}
\usepackage{hyperref}
\usepackage{xcolor}
\usepackage{graphicx}
\usepackage{subcaption}
\usepackage{booktabs}
\usepackage{mathtools}
\usepackage{braket}
\usepackage{amssymb}
\usepackage{multirow}
\usepackage{pifont}
\newcommand{\cmark}{\ding{51}} 
\newcommand{\xmark}{\ding{55}} 
\usepackage{amsthm}
\newtheorem{assumption}{Assumption}
\newtheorem{theorem}{Theorem}
\newtheorem{lemma}{Lemma}

\newtheorem{definition}{Definition}
\newtheorem{proposition}{Proposition}

\newtheorem{remark}{Remark}

\usepackage{listings}
\usepackage{quantikz}
\usepackage{subcaption}
\usepackage{caption}
\usepackage[T1]{fontenc}
\usepackage[utf8]{inputenc}

\definecolor{codebg}{RGB}{245,245,245}
\definecolor{codeframe}{RGB}{200,200,200}
\definecolor{codekw}{HTML}{1F77B4}
\definecolor{codestr}{HTML}{D62728}
\definecolor{codecm}{HTML}{2CA02C}

\lstdefinestyle{cmabcode}{
  backgroundcolor=\color{codebg},
  frame=single, rulecolor=\color{codeframe}, frameround=tttt,
  numbers=left, numberstyle=\tiny\color{gray}, numbersep=6pt,
  basicstyle=\ttfamily\footnotesize, columns=fullflexible,
  keywordstyle=\bfseries\color{codekw},
  commentstyle=\itshape\color{codecm!90},
  stringstyle=\color{codestr},
  showstringspaces=false, upquote=true, keepspaces=true,
  breaklines=true, extendedchars=false
}
\begin{document}
\title{Quantum Noise Mitigation with Adaptive Zero-Noise Extrapolation: A Contextual Multi-Armed Bandits Approach}

\author{Ratun Rahman and Dinh C. Nguyen
\thanks{Ratun Rahman and Dinh C Nguyen are with the Department of Electrical and Computer Engineering, University of Alabama in Huntsville, Huntsville, AL 35899, USA, emails: rr0110@uah.edu, dinh.nguyen@uah.edu}
}

\markboth{accepted in the IEEE Journal on Selected Areas in Communications}%
{Shell \MakeLowercase{\textit{et al.}}: Bare Demo of IEEEtran.cls for IEEE Journals}

\maketitle

\begin{abstract}
Variational quantum circuits (VQCs) are central to near-term quantum computing, yet their practical deployment is severely hindered by noise. While existing error mitigation methods, such as zero-noise extrapolation (ZNE), typically assume static noise, real noisy intermediate-scale quantum (NISQ) systems exhibit dynamic, time-varying noise that remains largely unaddressed. To overcome this critical gap, our work introduces a novel adaptive noise mitigation framework for VQCs that integrates ZNE with \textcolor{black}{contextual multi-armed bandits (CMAB)}, enabling dynamic, context-aware selection of circuit-folding levels based on ansatz parameters (e.g., depth, parameter count) and the evolving noise environment. Unlike fixed-fold or heuristic ZNE, our approach uses online adaptation to improve the accuracy of ZNE and reduce redundant quantum circuit executions. Our extensive simulations and experiments on real quantum hardware reveal the following important properties: (i) deeper VQCs accumulate noise, degrading accuracy and increasing the number of quantum circuit executions; (ii) ZNE restores estimator fidelity when the folding level is chosen appropriately; and (iii) CMAB-guided folding cuts quantum circuit execution round trips by up to 40\%, bytes exchanged by up to 35\%, and end-to-end cost by up to 30\% under a 10~Mbps budget, with up to 6.9\% higher estimator fidelity (CIFAR-10, depth 3, noise band $\eta=0.05$), versus fixed-fold and grid-search ZNE. These results demonstrate substantial performance gains over existing noise mitigation methods, underscoring the effectiveness of our design in supporting robust noise mitigation for VQCs. The source code is also publicly released to support reproducibility.

\end{abstract}


\IEEEpeerreviewmaketitle

\section{Introduction}
Variational quantum circuits (VQCs), foundational to near-term variational algorithms, offer advantages in optimization, simulations, and machine learning tasks~\cite{biamonte2017quantum, rahman2026quantum}. However, in noisy intermediate-scale quantum (NISQ) devices, \textit{decoherence} and \textit{gate-level errors} severely limit fidelity~\cite{preskill2018quantum, rahman2025sporadic}. Dominant noise mechanisms, amplitude damping, dephasing, and thermal excitations, alter quantum states, with effects compounding in deeper, more entangled circuits. Most error-mitigation research assumes static noise models, yet practical NISQ systems exhibit dynamic, time-varying noise due to thermal fluctuations, crosstalk, and calibration drift, which is underexplored~\cite{cai2023quantum, koenig2025adaptive}. Gate infidelities scale with circuit depth~\cite{murali2019full}, dephasing disrupts phase coherence critical for interference, and thermal processes induce stochastic population transitions~\cite{gambetta2017building}, all of which degrade performance in superconducting qubits due to cryogenic constraints and imperfect isolation. 


The VQCs in NISQ devices have motivated a variety of error-mitigation strategies, among which zero-noise extrapolation (ZNE) is considered the most appealing, as it is hardware-agnostic and easily compatible with existing processes~\cite{temme2017error}. ZNE works by intentionally increasing noise, typically through unitary folding of gates that keep the ideal circuit unitary, and then extrapolating observed observables to the zero-noise limit~\cite{giurgica2020digital}. However, in reality, the requested fold factor $\lambda$ does not correspond to a consistent increase in effective noise: the actual amplification varies with time, qubit placement and routing, crosstalk patterns, compiler/scheduler selections, and calibration cycles. Consequently, offline folding level selection generates model mismatches, which bias the extrapolation and cause redundant circuit re-executions~\cite{strikis2021learning}. Dynamic noise, such as thermal populations, dephasing rates, and gate infidelities, can invalidate fixed-fold schedules and reduce robustness~\cite{koenig2025adaptive}. These constraints highlight the need for an \textit{adaptive, context-aware mechanism} that adjusts noise scaling online, guided by circuit-level descriptors (e.g., depth, entangling density) and device-level signals (e.g., recent calibration/telemetry), to recover fidelity while minimizing superfluous evaluations.


\subsection{Related Works}\label{sec: related}
\begin{table*}[t]
\centering
\color{black}
\scriptsize
\setlength{\tabcolsep}{3pt}
\renewcommand{\arraystretch}{1.15}
\caption{\footnotesize {\color{black}Comparison of representative quantum error-mitigation and adaptive decision-making approaches. The table reports shot overhead, adaptation capability, communication-cost awareness, contextual decision support, theoretical characterization, and the main advantage and disadvantage of each method. Here,} \cmark=Yes, \xmark=No, $\sim$=Partial{\color{black}, and theory codes are E=Empirical, P=Probabilistic, S=Statistical, and SC=Sample Complexity.}}
\label{tab:comparison}
\resizebox{\textwidth}{!}{%
\begin{tabular}{p{2.35cm}ccccc p{4.2cm} p{4.6cm}}
\toprule
Approach 
& Shot Overhead 
& Adaptation 
& Communication 
& Context 
& Theory 
& Advantage 
& Disadvantage \\
\midrule
CDR~\cite{strikis2021learning}          
& Mod. & \xmark & \xmark & \xmark & E 
& Learns correction from classically tractable Clifford or near-Clifford circuits.
& Depends on representative training circuits and may not generalize to target VQCs. \\

VD~\cite{huggins2021virtual}            
& High & \xmark & \xmark & \xmark & S 
& Suppresses incoherent errors by projecting noisy states toward the dominant eigenstate.
& Requires multiple noisy state copies, increasing sampling and execution overhead. \\

PEC~\cite{cai2023quantum}               
& High & \xmark & \xmark & \xmark & P 
& Can provide unbiased estimates when the noise model is accurately characterized.
& Often suffers from variance amplification and large shot complexity. \\

RC~\cite{jain2023improved}              
& Mod. & \xmark & \xmark & \xmark & S 
& Transforms coherent errors into more stochastic noise.
& Does not directly optimize ZNE scale selection or communication cost. \\

RL~\cite{bordoni2024quantum}            
& High & \cmark & \xmark & $\sim$ & E 
& Learns adaptive mitigation or control policies from feedback.
& May require many training interactions and repeated evaluations. \\

Hardware-aware optimization~\cite{niu2020hardware}   
& Mod. & \cmark & $\sim$ & \xmark & E 
& Improves execution through device-aware mapping, routing, or scheduling.
& Is often device-specific and does not directly address adaptive ZNE fold selection. \\
\midrule

ZNE (standard)~\cite{temme2017error}        
& High & \xmark & \xmark & \xmark & E 
& Hardware-agnostic and requires no additional logical qubits.
& Uses fixed scale factors and may require redundant circuit executions. \\

ZNE (best-practice)~\cite{majumdar2023best}      
& High & \xmark & \xmark & \xmark & E 
& Improves extrapolation reliability through practical noise-regularization choices.
& Still depends on preselected or manually tuned scale sets. \\

ZNE (noise-aware)~\cite{hour2024improving} 
& High & \xmark & \xmark & $\sim$ & E 
& Uses calibration or noise information to guide mitigation.
& Does not learn online fold selection with explicit communication-cost control. \\

ZNE (global folding)~\cite{sohn2025application} 
& High & \xmark & \xmark & \xmark & E 
& Amplifies noise systematically while preserving the target unitary. 
& Can increase circuit depth, latency, and execution overhead. \\

Adaptive KIK~\cite{koenig2025adaptive}  
& Mod. & \cmark & \xmark & $\sim$ & S 
& Improves robustness under time-varying noise.
& Does not explicitly optimize QPU--CPU communication cost.\\
\midrule

\textbf{CMAB--ZNE (ours)}               
& Low & \cmark & \cmark & \cmark & SC 
& Adaptively selects ZNE folding arms using circuit and noise context.
& Requires reliable context/reward estimation and initial exploration. \\
\bottomrule
\end{tabular}%
}
\end{table*}

\subsubsection{Variational Quantum Circuits under Noise}
In order to take advantage of expressivity and inductive bias for learning tasks, VQCs in quantum machine learning (QML) combine quantum and classical processing~\cite{biamonte2017quantum,huggins2021virtual}. Although hybrid VQC architectures are promising for classification and regression, NISQ devices suffer from decoherence, gate/readout errors, and calibration drift, which degrade predictive performance and increase the number of circuit evaluations needed for stable estimates~\cite{mari2021extending,wang2021noise}. In edge-cloud applications, these hardware constraints translate directly into communication overhead: each additional quantum run results in a quantum processing unit (QPU)-central processing unit (CPU) exchange, payload transfer, and latency accumulation.

Clifford data regression (CDR) uses Clifford-executable circuits to predict observables but requires Clifford data or its surrogates~\cite{strikis2021learning}. Virtual distillation (VD) reduces errors by projecting onto the dominating eigenstate of the noisy density matrix, but it requires many state copies and has a large sampling overhead~\cite{huggins2021virtual}. Probabilistic error cancellation (PEC) simulates noisy channels using quasi-probabilities; however, it can significantly increase shot complexity due to variance amplification~\cite{cai2023quantum}. Randomized compiling (RC) reduces coherence errors by adjusting circuits into noise-resilient ensembles but does not eliminate communication overheads~\cite{jain2023improved}. {\color{black}Although these techniques improve VQC reliability without full quantum error correction, they often require representative training circuits, additional state copies, precise noise characterization, or repeated executions. Table~\ref{tab:comparison} summarizes their key advantages and limitations.}

\subsubsection{Zero-Noise Extrapolation for Noise Mitigations in VQCs}
ZNE estimates zero-noise observables by executing intentionally noise-amplified circuit variations and extrapolating resulting measurements toward the zero-noise limit~\cite{temme2017error}. Recent best practices address coherent-error behavior and advocate Pauli twirling/whirling to regularize noise before extrapolation~\cite{majumdar2023best}. Calibration-aware folding uses hardware telemetry to inform scaling decisions~\cite{hour2024improving}. It applies to various platforms, including silicon spin qubits, employing global folding and inversion protocols to achieve high-state fidelity~\cite{sohn2025application}. Adaptive KIK improves resistance against time-varying noise by combining adaptive folding and statistical filtering~\cite{koenig2025adaptive}. {\color{black}Although ZNE is hardware-agnostic and practical for NISQ devices, most variants still rely on fixed, manually selected, or grid-searched folding scales, which can increase circuit depth, latency, shot count, and QPU--CPU communication under dynamic noise.}

\subsubsection{Multi-Armed and Contextual Bandits}
The multi-armed bandits (MABs) approach sequential decision-making by balancing exploration and exploitation with traditional guarantees. For instance, UCB1 achieves logarithmic regret in stochastic settings~\cite{auer2002finite}, whereas Thompson sampling has competitive theoretical and empirical performance~\cite{thompson1933likelihood,agrawal2012analysis}. Bandits are used in large-scale online systems to quickly choose and personalize models based on streaming feedback~\cite{li2010contextual}. NISQ reward signals can be non-stationary and high-variance due to time-varying decoherence and calibration drift, which makes fixed exploration schedules challenging. \textcolor{black}{Contextual multi-armed bandits (CMABs)} use side information (environmental/input features) to condition decisions, addressing heterogeneity and drift~\cite{li2010contextual,esmaeili2025robust}. This paradigm has proven to be effective in recommendation, dynamic pricing, and adaptive feature acquisition~\cite{xue2025multiple,bouneffouf2025multi}. Quantum-focused bandit formulations can guide quantum decision-making: quantum contextual bandits for Hamiltonian learning achieve near-optimal energy estimates~\cite{brahmachari2024quantum}; quantum linear-algebra tools can accelerate Thompson sampling~\cite{lumbreras2022multi}; and quantum-kernelized UCB provides regret guarantees~\cite{hikima2024quantum}. {\color{black}However, rather than cost-aware ZNE fold-set selection, current quantum bandit studies primarily concentrate on measurement selection, Hamiltonian learning, parameter tuning, or regret analysis.}

\begin{figure*}
    \centering
    \includegraphics[width=0.99\linewidth]{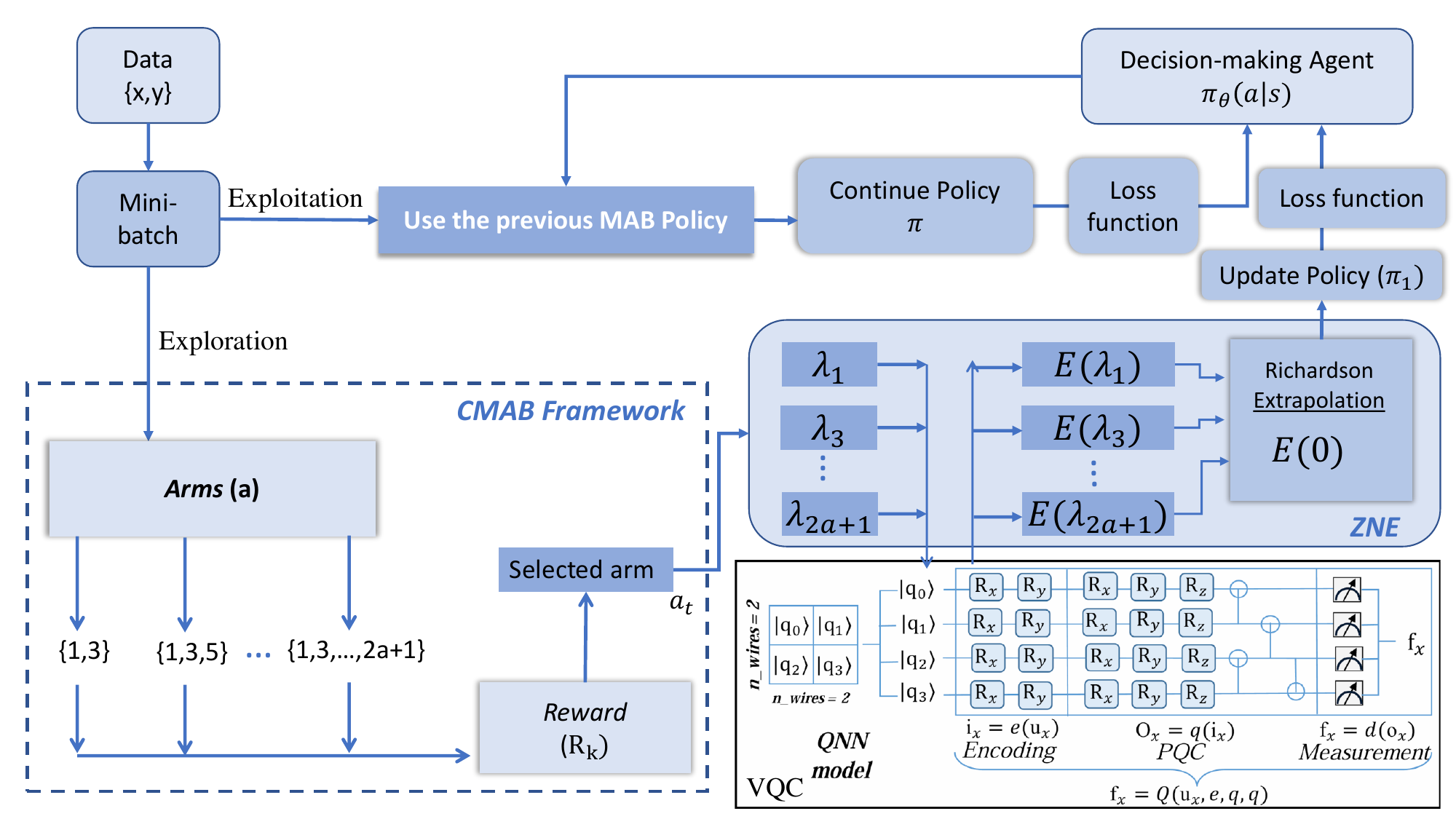}
    \caption{\small An overview of the proposed CMAB-ZNE training loop.  In each mini-batch, the agent either \emph{exploits} the current bandit policy or \emph{explores} candidate arms—sets of ZNE noise scales such as $\{1,3\}$, $\{1,3,5\}$, \ldots, $\{1,3,\ldots,2a\!+\!1\}$.  The chosen arm $a_t$ sets the folding scales $\lambda$ utilized by the ZNE module to evaluate expectations $E(\lambda)$ and extrapolate the zero-noise estimate $E(0)$ (e.g., using Richardson).  The variational quantum circuit (encoding $\rightarrow$ PQC $\rightarrow$ measurement) generates predictions $f_x$, which are used to compute the task loss.  A safety-aware reward $R_k$, which trades off accuracy/variance and circuit depth, is returned to update the contextual bandit policy $\pi_{\theta}(a\mid s)$ for the next batch, finishing the explore-exploit cycle.}
    \label{fig:overview}
    \vspace{-1em}
\end{figure*}

{\color{black}
\subsubsection{Positioning of Our Work}
Unlike prior mitigation techniques, our CMAB--ZNE framework treats ZNE fold-set selection as an online, cost-aware decision problem. The proposed policy adapts to the circuit and noise context while explicitly accounting for execution overhead, including round trips, exchanged bytes, and runtime cost.
}

\subsection{Motivations and Key Contributions}
Despite such research efforts, several limitations still exist:
\begin{itemize}
    \item Most noise mitigation frameworks \cite{strikis2021learning, cai2023quantum, temme2017error, koenig2025adaptive, bordoni2024quantum} do not explicitly optimize \textit{communication costs}, i.e., round trips/repeated quantum circuit executions to implement ZNE, and often lack real-time adaptivity to \emph{time-varying} noise in latency-sensitive settings.
   
    \item Existing advanced ZNE operations \cite{majumdar2023best, hour2024improving, sohn2025application, koenig2025adaptive} rely on fixed or heuristic scaling schedules, leading to redundant evaluations and additional quantum circuit execution round-trip without optimizing end-to-end communication cost. Moreover, they typically ignore \emph{context} (ansatz/noise features) when selecting folding levels online under an explicit evaluation budget.
     \item No CMAB works \cite{li2010contextual, esmaeili2025robust, xue2025multiple, bouneffouf2025multi, brahmachari2024quantum, lumbreras2022multi, hikima2024quantum} go beyond state or measurement selection and theoretical analysis; in particular, none perform contextual fold selection for ZNE. Table~\ref{tab:comparison} compares related error-mitigation methods with our work. 
\end{itemize}

\noindent \textbf{Motivated by the limitations, we propose a novel context-aware ZNE technique that transforms noise-scaling for VQCs into a CMAB problem.} Instead of preset folds, the agent learns the fold factor $\lambda$ online from the circuit and noise context, yielding a more favorable fidelity--efficiency trade-off, minimizing overfolding, and improving resilience for VQC workloads under drift. The CMAB policy balances exploration and exploitation in real time, and a budgeted-bandit variant terminates when a target estimator-variance threshold is met, resulting in equivalent fidelity with fewer quantum runs. A Lindblad model with time-varying rates captures noise dynamics~\cite{krinner2019engineering}, with a cost-controlled real-hardware validation illustrating feasibility. \textcolor{black}{Fig.~\ref{fig:overview} summarizes the overall CMAB--ZNE workflow. The figure connects the three main components of the proposed approach: the VQC generates noisy expectation values, the ZNE module estimates the zero-noise output using the selected folding arm, and the CMAB agent updates its policy based on the resulting reward. This visualization clarifies that the bandit does not replace ZNE; rather, it adaptively selects the ZNE folding configuration according to the current circuit and noise context.}
Our contributions are summarized as follows.
\begin{itemize}
    \item We introduce a dynamic noise model for \emph{VQCs} based on the Lindblad master equation, extracting gate- and circuit-level contextual features (e.g., depth-normalized infidelity, coherence budget) from time-varying decoherence and relaxation, augmented with public calibration metrics ($T_1/T_2$, gate/readout errors) (Section II).

    \item We formulate ZNE fold selection as a CMAB problem, using a context vector combining Lindblad-derived features and calibration signals, with a budgeted-bandit stopping rule to minimize evaluations when target variance is reached, and analyze its sample complexity (Section II).

    \item We conduct extensive simulations and experiments on real quantum hardware via VQC-based machine learning tasks on   \textsc{CIFAR-10} and \textsc{EuroSAT} datasets (Section III), showing \textit{CMAB--ZNE} achieves up to 6.9\% higher fidelity (CIFAR-10, depth 3, noise band $\eta=0.05$) and reduces quantum circuit execution round trips by up to 40\%, bytes by 35\%, versus fixed-fold and grid-search ZNE, and is competitive with other mitigation approaches, while preserving accuracy. \textbf{\textit{Our source code is online at: \url{https://github.com/Ratun11/cmab_zne/}.}}
\end{itemize}

\section{Methodology} \label{sec: method}
\subsection{Noise Model} 
\label{subsec:noise_model}

\noindent\textbf{Lindblad Dynamics.}
The qubit density matrix \(\rho(t)\) evolves under thermal relaxation, excitation, and dephasing as
\begin{equation}
\begin{aligned}
    \frac{d \rho}{d t} = 
&\gamma n_{\mathrm{noise}}(t) \mathcal{L}[\sigma_+](\rho) +
\gamma \left(n_{\mathrm{noise}}(t) + 1\right) \mathcal{L}[\sigma_-](\rho) \\
&+
\frac{1}{2 T_2} \mathcal{L}[\sigma_z](\rho),
\end{aligned}
\label{eq:lindblad_full}
\end{equation}
where \(\gamma\) is the emission rate, \(n_{\mathrm{noise}}(t)\) is the thermal photon number, \(T_2\) is the dephasing time, and \(\mathcal{L}[A](\rho)=A\rho A^\dagger-\frac{1}{2}\{A^\dagger A,\rho\}\) is the Lindblad super-operator~\cite{alicki2007quantum}.
{\color{black}
By contrasting noiseless control, stochastic control perturbation, and Lindblad-style relaxation/dephasing, Fig.~\ref{fig:noise_profiles} shows why a time-dependent noise model is required. The clean profile depicts ideal unitary evolution, whereas the noisy and Lindblad-driven profiles show how fluctuations, relaxation, and dephasing disturb the intended dynamics. In Eq.~\eqref{eq:lindblad_full}, the relaxation and thermal-excitation terms are associated with $T_1$, while the phase-decoherence term is associated with $T_2$. As a result, the model captures the dominant superconducting-qubit decoherence mechanisms while approximating slow non-stationary effects such as temperature drift, calibration variation, and device instability. These time-varying features motivate their inclusion in the CMAB context vector, since the folding decision should depend on both circuit depth and the current noise condition.
}

\begin{figure*}[t]
    \centering
    \includegraphics[width=0.85\textwidth]{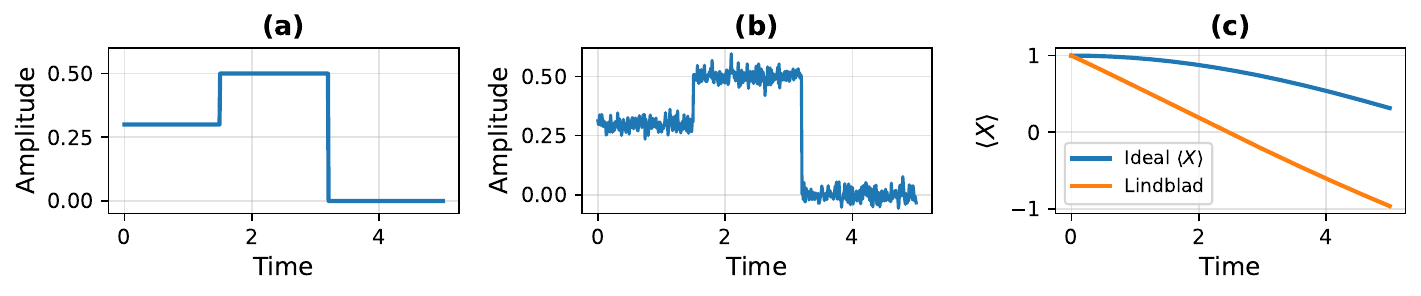}
    \caption{Quantum gate-noise visualization. The panels show: (a) clean control evolution, (b) noisy evolution with stochastic fluctuations, and (c) Lindblad-style relaxation and dephasing behavior.}
    \label{fig:noise_profiles}
\end{figure*}

\noindent\textbf{Time-dependent Thermal Noise.}
To capture fluctuating environments and non-stationary thermal loads caused by scheduling, drift, or cross-talk, we define
\begin{equation}
n_{\mathrm{noise}}(t) =
\frac{A - 1}{A} n_{\mathrm{BE}}(T_{\mathrm{qb}}(t))
+
\frac{1}{A} n_{\mathrm{BE}}(T_{\mathrm{ext}}(t)),
\label{eq:nnoise_dynamic}
\end{equation}
where
\[
n_{\mathrm{BE}}(T)=
\frac{1}{\exp\left(\frac{\hbar\omega_0}{k_BT}\right)-1}
\]
is the Bose--Einstein photon distribution, \(T_{\mathrm{qb}}(t)\) and \(T_{\mathrm{ext}}(t)\) are the time-varying temperatures of the qubit and external electronics, and \(A\) is the control-line attenuation factor~\cite{krinner2019engineering}. This formulation is consistent with empirical reports of temperature drift and time-correlated errors in NISQ hardware~\cite{hertzberg2021laser}.

\noindent\textbf{Performance Metrics.}
The time-dependent gate infidelity and worst-case error probability are
\begin{equation}
\begin{aligned}
    I_{\mathrm{F}}(t)&=\gamma\tau_{\mathrm{gate}}\left(1+n_{\mathrm{noise}}(t)\right),\\
p_{\mathrm{err}}(t)&=
\frac{1}{2}\gamma\tau_{\mathrm{gate}}
\left(1+2n_{\mathrm{noise}}(t)\right),
\end{aligned}
\end{equation}
where \(I_{\mathrm{F}}(t)\) measures deviation from the ideal unitary gate, \(p_{\mathrm{err}}(t)\) represents the worst-case error probability, and \(\tau_{\mathrm{gate}}\) is the gate duration.

{\color{black}
\noindent\textbf{Gate-Level Noise Impact.}
The gate-level expression used in this work follows from the weak-noise, short-gate-time approximation of the Lindblad/\textcolor{black}{metric-noise-resource} (MNR) model. Over a gate interval of duration $\tau_{\mathrm{gate}}$, the first-order accumulated error is proportional to the effective relaxation rate and thermal occupation. Neglecting higher-order terms $\mathcal{O}((\gamma\tau_{\mathrm{gate}})^2)$ under $\gamma\tau_{\mathrm{gate}}\ll 1$, the gate fidelity becomes
\begin{equation}
    M_{\mathrm{gate}}
    =
    1-I_F(t)
    \approx
    1-\gamma\tau_{\mathrm{gate}}\bigl(1+n_{\mathrm{noise}}(t)\bigr),
    \label{eq:metric}
\end{equation}
following the first-order Lindblad/MNR approximation for Markovian decoherence~\cite{alicki2007quantum,fellous2023optimizing}. This shows that gate fidelity decreases approximately linearly with gate duration and effective thermal occupation.
}

\noindent\textbf{Circuit-Level Impact.}
For a circuit with \(N_g\) gates, the instantaneous circuit fidelity is approximated as
\begin{equation}
M_{\mathrm{circuit}}
=
1-N_g I_{\mathrm{F}}(t)
=
1-N_g\gamma\tau_{\mathrm{gate}}
\left(1+n_{\mathrm{noise}}(t)\right).
\label{eq:mcircuit}
\end{equation}
Over an execution window \([0,T]\), the average circuit fidelity becomes
\begin{equation}
M_{\mathrm{circuit}}
=
1-\frac{N_g\gamma\tau_{\mathrm{gate}}}{T}
\int_0^T
\left(1+n_{\mathrm{noise}}(t)\right)dt .
\label{eq:mcircuit_dynamic}
\end{equation}
This form generalizes the static model by explicitly accounting for time-varying degradation during circuit execution.
{\color{black}
Fig.~\ref{fig:noise_effect} shows the physical interpretation of the model at the qubit-state level using QuTiP~\cite{johansson2012qutip}. In the ideal case, the qubit maintains coherence on the Bloch sphere; under relaxation and dephasing, the state becomes mixed, and the measured expectation values become biased. ZNE aims to remove this bias, while the CMAB policy determines how aggressively the circuit should be folded to estimate the zero-noise value.
}

\begin{figure}
    \centering
    \includegraphics[width=0.90\linewidth]{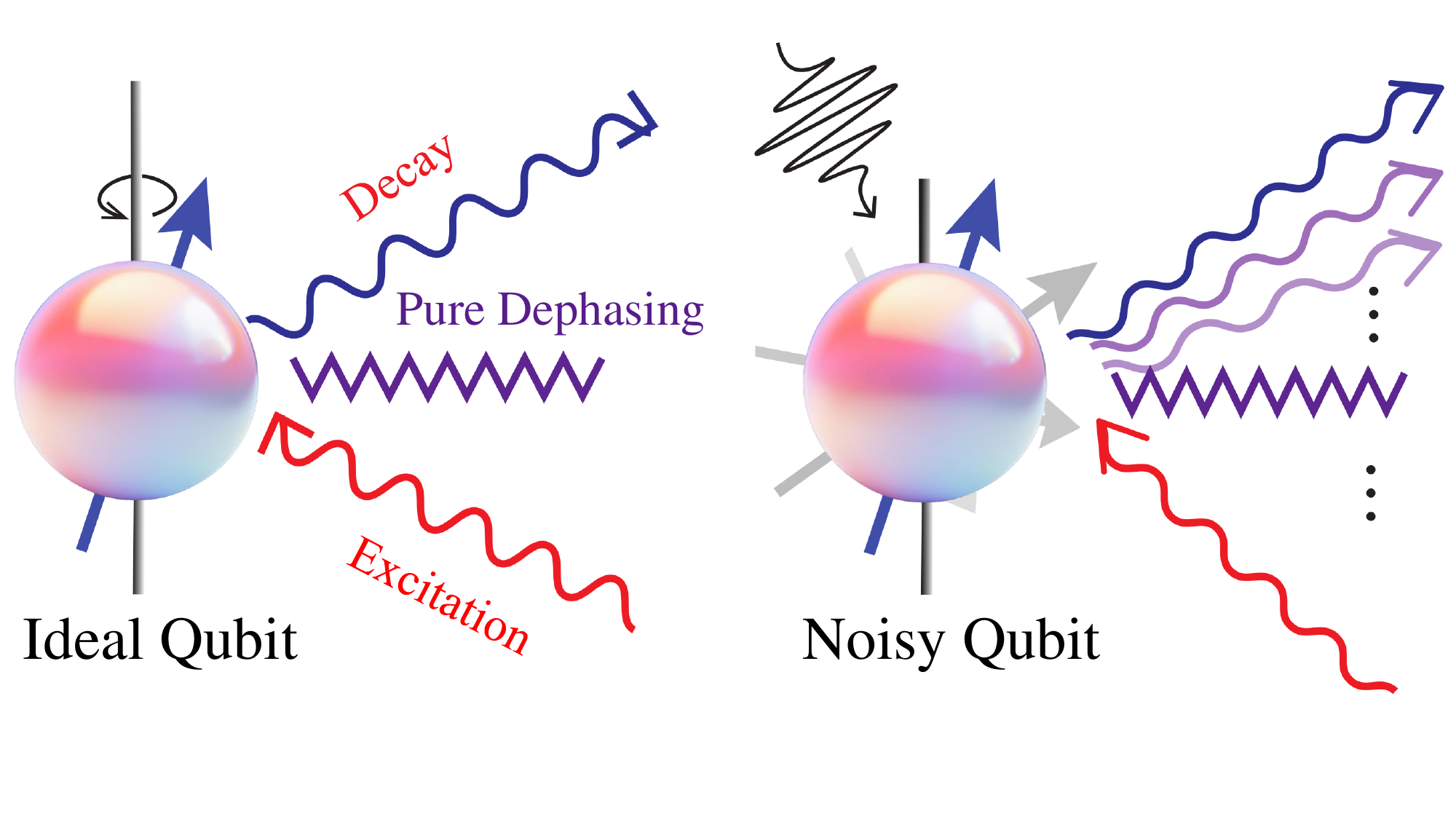}
    \vspace{-2em}
    \caption{\footnotesize Quantum noise impacts a single qubit. \textit{Left.} An ideal qubit maintains coherence. \textit{Right.} Environmental coupling causes relaxation/excitation (\(T_1\)) and pure dephasing (\(T_2\)), leading to a mixed state, biased measurements, and degraded learning performance.}
    \vspace{-1.5em}
    \label{fig:noise_effect}
\end{figure}

\subsection{ZNE for Noise Mitigation}
After characterizing gate-level noise using the MNR model~\cite{fellous2023optimizing}, we use \textit{zero-noise extrapolation} (ZNE) to estimate the ideal circuit output without modifying the hardware. ZNE intentionally amplifies noise through controlled circuit folding, evaluates the circuit at multiple noise scales, and extrapolates the measured observable back to the zero-noise limit~\cite{preskill2018quantum,temme2017error}.

\textbf{Motivation.} From the MNR model, the gate-level fidelity is
\[
M_{\mathrm{gate}} = 1 - \gamma \tau_{\mathrm{gate}} \left(1 + n_{\mathrm{noise}}(t) \right),
\label{eq:gate_fidelity}
\]
which shows that gate performance degrades with longer gate duration and stronger time-varying noise. ZNE exploits this behavior by treating noise as a controllable signal: instead of requiring exact hardware calibration, it evaluates the circuit under amplified noise and estimates the zero-noise observable through extrapolation. This improves output fidelity at the cost of additional circuit executions and increased effective depth~\cite{he2020zero}.

\textbf{Circuit Folding and Noise Scaling.} We implement ZNE using unitary circuit folding. For a circuit represented by unitary \(U\), the folded circuit is
\begin{equation}
U_{\mathrm{scaled}} = U (U^\dagger U)^n,
\label{eq:circuit_folding}
\end{equation}
where \(n\in\mathbb{Z}^{+}\) is the folding level. Since \((U^\dagger U)^nU=U\) ideally, folding preserves the target unitary while increasing the effective circuit depth and noise exposure. The corresponding noise scale is
\[
\lambda = 1+2n,
\]
and for a base circuit with \(N_g\) gates, the folded circuit has
\begin{equation}
N_g^{(\lambda)} = \lambda N_g = (1+2n)N_g .
\end{equation}
Substituting this into the time-varying MNR fidelity model gives
\begin{equation}
\begin{aligned}
M_{\mathrm{circuit}}
&= 1 - \frac{N_g^{(\lambda)}\gamma\tau_{\mathrm{gate}}}{T}
\int_0^T \left(1+n_{\mathrm{noise}}(t)\right)dt  \\
&= 1 - \frac{(1+2n)N_g\gamma\tau_{\mathrm{gate}}}{T}
\int_0^T \left(1+n_{\mathrm{noise}}(t)\right)dt .
\end{aligned}
\label{eq:circuit_fidelity_scaled}
\end{equation}
Thus, increasing \(\lambda\) amplifies noise in a controlled manner by increasing the effective number of gates~\cite{giurgica2020digital}.
\begin{figure}[t]
    \centering
    \includegraphics[width=0.7\columnwidth]{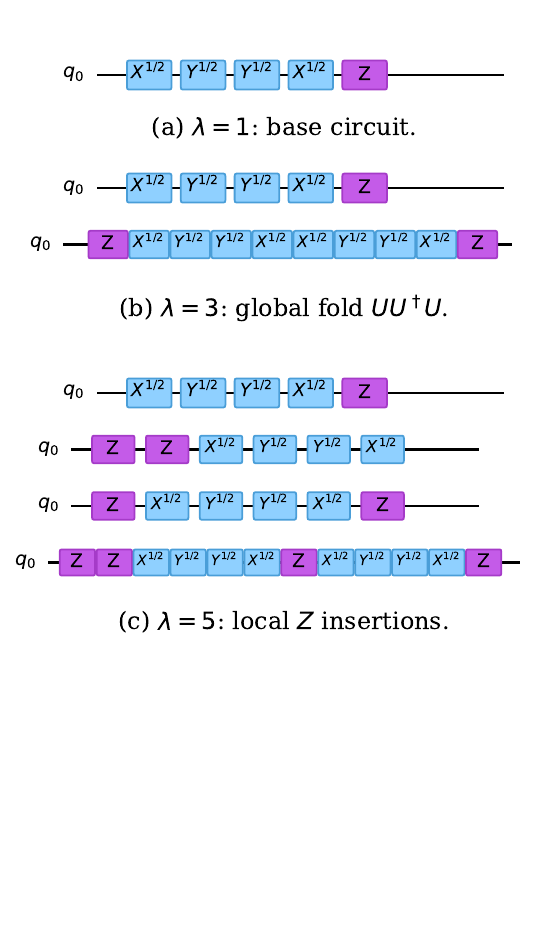}
    \vspace{-10em}
    \caption{\footnotesize Unitary folding for noise scaling in zero-noise extrapolation (ZNE). Larger folding factors increase effective circuit depth and noise exposure while preserving the ideal unitary operation.}
    \label{fig:zne-folds}
    \vspace{-0.8em}
\end{figure}
{\color{black} Fig.~\ref{fig:zne-folds} shows the circuit-level meaning of the folding factor $\lambda$. The base circuit corresponds to $\lambda=1$, while larger scales, such as $\lambda=3$ and $\lambda=5$, introduce logically canceling circuit segments that preserve the ideal unitary while increasing noise exposure. In CMAB--ZNE, each bandit arm corresponds to such a folding-scale set; therefore, larger folds may improve extrapolation but also increase circuit depth, runtime, and communication overhead.}

\textbf{Expectation Values and Extrapolation.} For each folded circuit, we compute an expectation value \(E(\lambda)\). Under a smooth noise response, this value can be approximated as
\begin{equation}
E(\lambda) = E(0) + c_1\lambda + c_2\lambda^2 + \mathcal{O}(\lambda^3),
\label{eq:expectation_polynomial}
\end{equation}
where \(E(0)\) is the ideal zero-noise expectation and \(c_1,c_2,\ldots\) capture noise-induced bias. With two scaling factors \(\lambda_1\) and \(\lambda_2\), first-order Richardson extrapolation estimates
\begin{equation}
E(0) \approx
\frac{\lambda_2E(\lambda_1)-\lambda_1E(\lambda_2)}
{\lambda_2-\lambda_1}.
\label{eq:richardson}
\end{equation}
For more than two scales, we fit a low-order polynomial
\begin{equation}
E(\lambda)=E(0)+\sum_{k=1}^{m}c_k\lambda^k,
\end{equation}
where the intercept gives the zero-noise estimate. In this work, we use linear fitting for \(\{1,3\}\) and low-order polynomial fitting for larger scale sets such as \(\{1,3,5\}\) and \(\{1,3,5,7\}\), balancing extrapolation accuracy against additional circuit-execution cost.

\begin{figure}[t]
\centering
\includegraphics[width=0.99\linewidth]{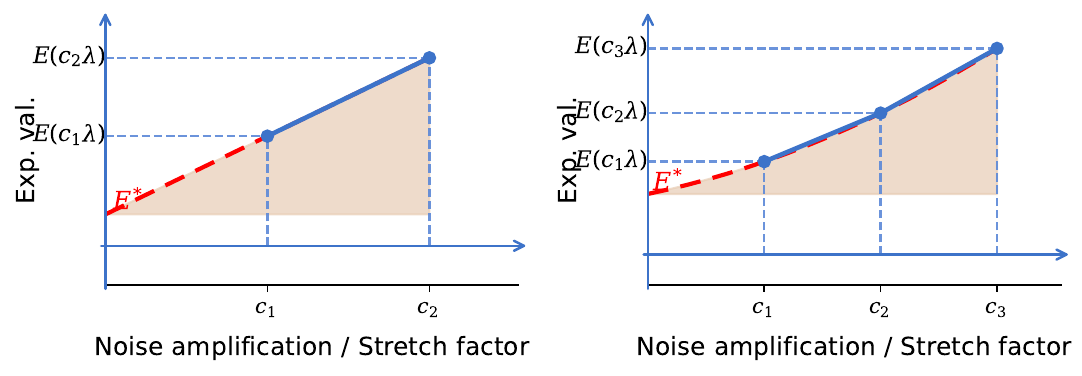}
\caption{\footnotesize 
Illustrating ZNE with varying noise amplification levels. \textit{Left:} Linear extrapolation from two amplified circuits at \(c_1\lambda\) and \(c_2\lambda\) predicts the noise-free expectation \(E^*\). \textit{Right:} Polynomial extrapolation using three points enables higher-order fitting.}
\vspace{-1em}
\label{fig:zne_extrapolation}
\end{figure}

{\color{black}  Fig.~\ref{fig:zne_extrapolation} illustrates the geometric interpretation of the ZNE extrapolation step. Each point corresponds to an expectation value measured from a folded circuit at a different noise scale, and the fitted curve is extrapolated back to the zero-noise intercept $E(0)$. This directly motivates the CMAB formulation, since each arm selects a different scale set and therefore changes both extrapolation quality and execution cost.}
The ZNE procedure used in our framework is therefore summarized as follows. For each selected arm \(\Lambda\), we fold the base circuit at every \(\lambda\in\Lambda\), execute each folded circuit with \(S\) shots, estimate \(E(\lambda)\), and fit the measured pairs \(\{(\lambda,E(\lambda))\}_{\lambda\in\Lambda}\) using Richardson or low-order polynomial extrapolation. The resulting zero-noise estimate \(\widehat{E}(0)\) is then passed to the downstream CMAB policy for reward computation and arm selection.

\textbf{Why ZNE Works.} ZNE works when the observable changes smoothly with the effective noise scale. Circuit folding amplifies the noise while preserving the ideal unitary, and Richardson or polynomial extrapolation cancels the leading noise terms to recover an estimate of \(E(0)\). Since ZNE is hardware-agnostic and does not require full noise calibration, it is suitable for NISQ devices with complex and time-varying noise~\cite{endo2018practical}.

\subsection{CMAB for Selecting ZNE Depth}
\label{subsec:zne-cmab}

We formulate ZNE fold selection as a cost-aware contextual multi-armed bandit (CMAB) problem. At each decision round \(t\), e.g., per epoch, mini-batch, or calibration interval, the policy observes a context vector \(x_t\in\mathbb{R}^p\), selects one ZNE arm, and receives a reward that reflects both extrapolation quality and execution cost under time-varying MNR noise.

\textbf{Arms and cost.}
Each arm \(k\in\{1,\ldots,K\}\) represents a ZNE folding-scale set
\[
\Lambda_k=\{\lambda_1^k,\ldots,\lambda_{m_k}^k\},
\]
such as \(\{1,3\}\), \(\{1,3,5\}\), or \(\{1,3,5,7\}\). Larger arms provide more extrapolation points but require more folded-circuit executions. Since folding with scale \(\lambda\) approximately multiplies the effective gate count and runtime,
\begin{equation}
N_g(\lambda)=\lambda N_g,\qquad
\tau_{\mathrm{circ}}(\lambda)=\lambda N_g\tau_{\mathrm{gate}},
\label{eq:fold-scaling}
\end{equation}
we define the cumulative scale of arm \(k\) as
\begin{equation}
S_k=\sum_{j=1}^{m_k}\lambda_j^k,
\qquad
\deg(\Lambda_k)=m_k-1 .
\label{eq:cum-scale}
\end{equation}
A ZNE evaluation with arm \(k\) induces \(R_k=m_k\) round trips, approximately \(B_k\approx m_k(B_{\mathrm{req}}+B_{\mathrm{resp}})\) transferred bytes, and runtime cost \(C_k^{\mathrm{rt}}\propto N_g\tau_{\mathrm{gate}}S_k\). We combine these terms using the normalized cost proxy
\begin{equation}
\tilde C_k=
\frac{\alpha R_k+\beta B_k+\gamma c C_k^{\mathrm{rt}}}
{\max_j(\alpha R_j+\beta B_j+\gamma c C_j^{\mathrm{rt}})}
\in[0,1],
\label{eq:combined-cost}
\end{equation}
where \(c\) converts runtime into the same scale and \(\alpha,\beta,\gamma\ge0\) weight round trips, bytes, and runtime.

\textbf{Context and reward.}
The context \(x_t\) summarizes the current circuit and noise condition using MNR/Lindblad features, circuit descriptors, recent performance deltas, and available hardware proxies such as calibration or drift indicators. After selecting arm \(k_t\), the base circuit is folded at each \(\lambda\in\Lambda_{k_t}\), executed with \(S\) shots, and used to estimate the measured pairs
\[
\{(\lambda,E(\lambda))\}_{\lambda\in\Lambda_{k_t}} .
\]
These pairs are fitted using Richardson or low-order polynomial extrapolation to obtain the zero-noise estimate \(\widehat{E}_{t,k_t}(0)\).

The selected arm receives a scalar reward that balances prediction quality, extrapolation stability, and folding overhead:
\begin{equation}
R_{t,k_t}
=
-\mathrm{CE}_{t,k_t}^{\mathrm{ZNE}}
-\beta\,\mathrm{Var}\!\left[\widehat{E}_{t,k_t}(0)\right]
-\alpha\left(d(\Lambda_{k_t})-d_0\right)^2 .
\label{eq:reward_detailed_cmab_zne}
\end{equation}
Here, \(\mathrm{CE}_{t,k_t}^{\mathrm{ZNE}}\) is the cross-entropy loss from the ZNE-corrected prediction, \(\mathrm{Var}[\widehat{E}_{t,k_t}(0)]\) measures extrapolation uncertainty, and \(d(\Lambda_{k_t})\) is a folding-depth proxy such as \(\max_{\lambda\in\Lambda_{k_t}}\lambda\) or \(\sum_{\lambda\in\Lambda_{k_t}}\lambda\). Thus, CMAB--ZNE favors shallow arms when they are accurate and stable, while selecting deeper arms only when their extrapolation gain justifies the additional cost.

\textbf{Cost-aware LinUCB selection.}
For each arm \(k\), we maintain ridge-regression parameters
\[
A_k\leftarrow \lambda_0 I_p,\qquad b_k\leftarrow 0,\qquad
\widehat{\theta}_k=A_k^{-1}b_k .
\]
Given \(x_t\), the upper-confidence score is
\begin{equation}
\mathrm{UCB}_k(x_t)
=
x_t^\top\widehat{\theta}_k
+
\eta\sqrt{x_t^\top A_k^{-1}x_t},
\end{equation}
and the selected arm is
\begin{equation}
k_t=
\arg\max_k
\left(
\mathrm{UCB}_k(x_t)-\lambda_c\tilde C_k
\right),
\label{eq:select}
\end{equation}
where \(\eta\) controls exploration and \(\lambda_c\) controls cost sensitivity. After observing the reward, only the selected arm is updated:
\begin{equation}
A_{k_t}\leftarrow \rho A_{k_t}+x_tx_t^\top,\quad
b_{k_t}\leftarrow \rho b_{k_t}+R_{t,k_t}x_t,\quad
\widehat{\theta}_{k_t}\leftarrow A_{k_t}^{-1}b_{k_t}.
\label{eq:discounted_ridge_update_cmab_zne}
\end{equation}
The discount factor \(\rho<1\) reduces the influence of stale observations under dynamic noise, while \(\rho=1\) recovers the non-discounted update analyzed in Section~\ref{subsec:theory_cmab_zne}. Ridge regularization keeps \(A_k\) positive definite, and the exploration bonus prevents premature convergence to a single arm.

\textbf{Budgeted stopping.}
To limit communication, CMAB--ZNE stops additional ZNE queries when the selected estimate is stable or when the cost-adjusted gain over the next-best arm is negligible:
\begin{equation}
\mathrm{Var}\!\left[E_{t,k_t}^{\ast}\right]\le\sigma_{\mathrm{target}}^2
\quad\text{or}\quad
\Delta_{\mathrm{UCB}}(t)<\epsilon,
\label{eq:budget}
\end{equation}
where
\(
\Delta_{\mathrm{UCB}}(t)=
\left(\mathrm{UCB}_{k_t}(x_t)-\lambda_c\tilde C_{k_t}\right)
-
\max_{j\ne k_t}
\left(\mathrm{UCB}_{j}(x_t)-\lambda_c\tilde C_j\right).
\)
This aligns extrapolation quality with round-trip, byte-transfer, and runtime constraints.

\begin{algorithm}[t]
\footnotesize
\caption{\footnotesize CMAB-ZNE: Cost-Aware Contextual Bandit for ZNE Arm Selection}
\label{alg:cmab_zne}
\begin{algorithmic}[1]
\REQUIRE Arms \(\{\Lambda_k\}_{k=1}^{K}\), context extractor \(x_t\), probe size \(B\), exploration \(\eta\), ridge \(\lambda_0\), cost penalty \(\lambda_c\), discount \(\rho\), reward weights \(\alpha,\beta\), preferred depth \(d_0\)
\ENSURE Learned parameters \(\{\widehat{\theta}_k\}\), decision log

\STATE Precompute normalized costs \(\tilde C_k\) for all arms.
\FOR{each arm \(k=1,\ldots,K\)}
    \STATE \(A_k\leftarrow \lambda_0 I_p\), \(b_k\leftarrow 0\), \(\widehat{\theta}_k\leftarrow A_k^{-1}b_k\)
\ENDFOR

\FOR{decision round \(t=1,\ldots,T\)}
    \STATE Observe context \(x_t\) from MNR features, circuit statistics, and recent feedback.
    \FOR{each arm \(k=1,\ldots,K\)}
        \STATE \(\mathrm{Score}_k \leftarrow x_t^\top\widehat{\theta}_k+\eta\sqrt{x_t^\top A_k^{-1}x_t}-\lambda_c\tilde C_k\)
    \ENDFOR
    \STATE \(k_t\leftarrow \arg\max_k \mathrm{Score}_k\)
    \STATE Sample probe batch \(\mathcal{D}_t\) of size \(B\).
    \STATE Run ZNE with \(\Lambda_{k_t}\) to obtain \(\widehat{E}_{t,k_t}(0)\) and task predictions.
    \STATE Compute \(\mathrm{CE}^{\mathrm{ZNE}}_{t,k_t}\), \(\mathrm{Var}[\widehat{E}_{t,k_t}(0)]\), and reward \(R_{t,k_t}\) using Eq.~\eqref{eq:reward_detailed_cmab_zne}.
    \STATE Update selected arm:
    \[
    A_{k_t}\leftarrow \rho A_{k_t}+x_tx_t^\top,\quad
    b_{k_t}\leftarrow \rho b_{k_t}+R_{t,k_t}x_t,\quad
    \widehat{\theta}_{k_t}\leftarrow A_{k_t}^{-1}b_{k_t}.
    \]
    \STATE Apply \(\Lambda_{k_t}\) to the main workload until the next decision round and log \((t,x_t,k_t,R_{t,k_t})\).
\ENDFOR
\STATE \textbf{return} \(\{\widehat{\theta}_k\}\), decision log.
\end{algorithmic}
\end{algorithm}

\textbf{Offline exploration and hardware exploitation.}
In offline replay, expectation-value dumps are available for all scales, so counterfactual rewards can be reconstructed for unselected arms and used for regret, oracle, and arm-selection analysis. In real hardware deployment, only the selected arm is evaluated on a small probe batch, and the learned policy applies that arm to the main workload until the next decision round.

{\color{black}
\subsection{Theoretical Analysis of Cost-Aware CMAB--ZNE Selection}
\label{subsec:theory_cmab_zne}

The following analysis focuses on the online fold-set selection component of CMAB--ZNE. It does not assume that ZNE completely removes all physical noise. Rather, it shows that when the reward signal reflects both extrapolation quality and execution cost, the proposed contextual bandit policy learns to select near-optimal ZNE folding arms with sublinear regret. For clarity, the regret analysis is given for the non-discounted LinUCB update, i.e., $\rho=1$ in Algorithm~1. In practice, the discounted case is used to improve adaptation under drift.

\begin{assumption}[Bounded context and reward]
\label{ass:bounded_context_reward}
The context vector satisfies $\|x_t\|_2\leq L$ for every decision round $t\in\{1,\ldots,T\}$. The observed reward satisfies $r_{t,k}\in[-1,1]$ for every ZNE arm $k\in\{1,\ldots,K\}$.
\end{assumption}

\begin{assumption}[Linear contextual reward model]
\label{ass:linear_reward}
For each arm $k$, the expected reward under context $x_t$ is
\begin{equation}
    \mathbb{E}[r_{t,k}\mid x_t]
    =
    x_t^{\top}\theta_k^{\star},
    \label{eq:linear_reward_model}
\end{equation}
where $\theta_k^{\star}\in\mathbb{R}^{p}$ is an unknown arm-specific parameter vector satisfying $\|\theta_k^{\star}\|_2\leq S$. The observed reward is
\begin{equation}
    r_{t,k}
    =
    x_t^{\top}\theta_k^{\star}
    +
    \epsilon_{t,k},
    \label{eq:reward_noise_model}
\end{equation}
where $\epsilon_{t,k}$ is conditionally zero-mean and $\sigma$-sub-Gaussian.
\end{assumption}

\begin{assumption}[Known normalized execution cost]
\label{ass:cost_aware_utility}
The normalized execution cost $\widetilde{C}_k\in[0,1]$ for each arm $k$ is known. At round $t$, the cost-aware utility of arm $k$ is
\begin{equation}
    u_{t,k}
    =
    x_t^{\top}\theta_k^{\star}
    -
    \lambda_c\widetilde{C}_k,
    \label{eq:cost_aware_utility}
\end{equation}
where $\lambda_c\geq0$ controls the execution-cost penalty.
\end{assumption}

\begin{definition}[Cost-aware regret]
\label{def:cost_aware_regret}
Let $k_t$ be the arm selected by CMAB--ZNE, and let
\begin{equation}
    k_t^{\star}
    =
    \arg\max_{k\in\{1,\ldots,K\}}
    \left(
    x_t^{\top}\theta_k^{\star}
    -
    \lambda_c\widetilde{C}_k
    \right)
\end{equation}
be the ideal cost-aware arm at round $t$. The cumulative cost-aware regret after $T$ rounds is
\begin{equation}
\begin{aligned}
    R_T
    =
    \sum_{t=1}^{T}
    \Big[
    &
    \left(
    x_t^{\top}\theta_{k_t^{\star}}^{\star}
    -
    \lambda_c\widetilde{C}_{k_t^{\star}}
    \right)
    -
    \left(
    x_t^{\top}\theta_{k_t}^{\star}
    -
    \lambda_c\widetilde{C}_{k_t}
    \right)
    \Big].
\end{aligned}
\label{eq:cost_aware_regret}
\end{equation}
\end{definition}

\begin{lemma}[Confidence bound for arm reward estimation]
\label{lem:confidence_bound}
For each arm $k$, define
\begin{equation}
    A_{t,k}
    =
    \lambda_0 I
    +
    \sum_{\tau<t:k_\tau=k}x_\tau x_\tau^{\top},
    \qquad
    b_{t,k}
    =
    \sum_{\tau<t:k_\tau=k}r_{\tau,k}x_\tau,
\end{equation}
where $\lambda_0>0$. The ridge estimator is
\begin{equation}
    \widehat{\theta}_{t,k}
    =
    A_{t,k}^{-1}b_{t,k}.
\end{equation}
Under Assumptions~\ref{ass:bounded_context_reward} and~\ref{ass:linear_reward}, with probability at least $1-\delta$, for every $t\in\{1,\ldots,T\}$ and all $k\in\{1,\ldots,K\}$,
\begin{equation}
    \left|
    x_t^{\top}
    \left(
    \widehat{\theta}_{t,k}
    -
    \theta_k^{\star}
    \right)
    \right|
    \leq
    \beta_T
    \sqrt{x_t^{\top}A_{t,k}^{-1}x_t},
    \label{eq:confidence_bound}
\end{equation}
where
\begin{equation}
    \beta_T
    =
    \sigma
    \sqrt{
    2\log
    \left(
    \frac{K}{\delta}
    \right)
    +
    p\log
    \left(
    1+\frac{T L^2}{\lambda_0 p}
    \right)
    }
    +
    \sqrt{\lambda_0}S .
    \label{eq:beta_T}
\end{equation}
\end{lemma}

\begin{proof}
Fix an arm $k$, and let $\mathcal{T}_{t,k}=\{\tau<t:k_\tau=k\}$ denote the set of previous rounds in which arm $k$ was selected. For each $\tau\in\mathcal{T}_{t,k}$, the reward model gives
\begin{equation}
    r_{\tau,k}
    =
    x_\tau^{\top}\theta_k^{\star}
    +
    \epsilon_{\tau,k}.
\end{equation}
Using the definition of $b_{t,k}$, we obtain
\begin{equation}
\begin{aligned}
    b_{t,k}
    &=
    \sum_{\tau\in\mathcal{T}_{t,k}}
    r_{\tau,k}x_\tau  \\
    &=
    \sum_{\tau\in\mathcal{T}_{t,k}}
    \left(
    x_\tau^{\top}\theta_k^{\star}
    +
    \epsilon_{\tau,k}
    \right)x_\tau  \\
    &=
    \left(
    \sum_{\tau\in\mathcal{T}_{t,k}}
    x_\tau x_\tau^{\top}
    \right)\theta_k^{\star}
    +
    \sum_{\tau\in\mathcal{T}_{t,k}}
    \epsilon_{\tau,k}x_\tau  \\
    &=
    \left(A_{t,k}-\lambda_0 I\right)\theta_k^{\star}
    +
    \sum_{\tau\in\mathcal{T}_{t,k}}
    \epsilon_{\tau,k}x_\tau .
\end{aligned}
\end{equation}
Therefore,
\begin{equation}
\begin{aligned}
    \widehat{\theta}_{t,k}-\theta_k^{\star}
    &=
    A_{t,k}^{-1}b_{t,k}
    -
    \theta_k^{\star}  \\
    &=
    A_{t,k}^{-1}
    \left(
    \sum_{\tau\in\mathcal{T}_{t,k}}
    \epsilon_{\tau,k}x_\tau
    -
    \lambda_0\theta_k^{\star}
    \right).
\end{aligned}
\end{equation}
Taking the $A_{t,k}$-weighted norm and applying the triangle inequality yields
\begin{equation}
\begin{aligned}
    \left\|
    \widehat{\theta}_{t,k}
    -
    \theta_k^{\star}
    \right\|_{A_{t,k}}
    &\leq
    \left\|
    \sum_{\tau\in\mathcal{T}_{t,k}}
    \epsilon_{\tau,k}x_\tau
    \right\|_{A_{t,k}^{-1}}
    +
    \sqrt{\lambda_0}
    \left\|
    \theta_k^{\star}
    \right\|_2 .
\end{aligned}
\end{equation}
Using a union bound over the $K$ arms and the self-normalized concentration inequality for linear bandits, with probability at least $1-\delta$,
\begin{equation}
    \left\|
    \sum_{\tau\in\mathcal{T}_{t,k}}
    \epsilon_{\tau,k}x_\tau
    \right\|_{A_{t,k}^{-1}}
    \leq
    \sigma
    \sqrt{
    2\log
    \left(
    \frac{K}{\delta}
    \right)
    +
    \log
    \left(
    \frac{\det(A_{t,k})}
    {\det(\lambda_0 I)}
    \right)
    } .
\end{equation}
Since $\|x_t\|_2\leq L$, the determinant ratio can be bounded as
\begin{equation}
    \log
    \left(
    \frac{\det(A_{t,k})}
    {\det(\lambda_0 I)}
    \right)
    \leq
    p\log
    \left(
    1+\frac{T L^2}{\lambda_0 p}
    \right).
\end{equation}
Using $\|\theta_k^{\star}\|_2\leq S$, we have
\begin{equation}
    \left\|
    \widehat{\theta}_{t,k}
    -
    \theta_k^{\star}
    \right\|_{A_{t,k}}
    \leq
    \beta_T .
\end{equation}
Finally, by Cauchy--Schwarz in the $A_{t,k}$ norm,
\begin{equation}
\begin{aligned}
    \left|
    x_t^{\top}
    \left(
    \widehat{\theta}_{t,k}
    -
    \theta_k^{\star}
    \right)
    \right|
    &\leq
    \|x_t\|_{A_{t,k}^{-1}}
    \left\|
    \widehat{\theta}_{t,k}
    -
    \theta_k^{\star}
    \right\|_{A_{t,k}} \\
    &\leq
    \beta_T
    \sqrt{x_t^{\top}A_{t,k}^{-1}x_t}.
\end{aligned}
\end{equation}
This proves the stated confidence bound.
\end{proof}

\begin{proposition}[Monotonic cost of nested ZNE arms]
\label{prop:zne_arm_cost}
Consider a nested family of ZNE arms
\begin{equation}
    \Lambda_k
    =
    \{\lambda_1^k,\ldots,\lambda_{m_k}^k\},
\end{equation}
where larger arms contain more folding scales and/or larger cumulative scale values. Define
\begin{equation}
    S_k
    =
    \sum_{j=1}^{m_k}\lambda_j^k .
\end{equation}
For fixed base circuit size $N_g$ and gate duration $\tau_{\mathrm{gate}}$, the runtime component of the ZNE execution cost is proportional to $S_k$. Furthermore, if two nested arms satisfy $m_i\leq m_j$ and $S_i\leq S_j$, then their round-trip, byte-transfer, and runtime costs satisfy
\begin{equation}
    R_i\leq R_j,\qquad
    B_i\leq B_j,\qquad
    C_i^{\mathrm{rt}}\leq C_j^{\mathrm{rt}} .
\end{equation}
\end{proposition}

\begin{proof}
Under unitary folding, a base circuit with $N_g$ gates is transformed into a folded circuit with approximately $\lambda N_g$ gates at scale $\lambda$. Therefore, executing all scales in arm $\Lambda_k$ requires the effective gate count
\begin{equation}
    N_{g,k}^{\mathrm{eff}}
    =
    \sum_{j=1}^{m_k}
    \lambda_j^k N_g
    =
    N_g S_k .
\end{equation}
For fixed gate duration $\tau_{\mathrm{gate}}$, the runtime component satisfies
\begin{equation}
    C_k^{\mathrm{rt}}
    \propto
    N_g\tau_{\mathrm{gate}}S_k .
\end{equation}
Consequently, if $S_i\leq S_j$, then $C_i^{\mathrm{rt}}\leq C_j^{\mathrm{rt}}$. In addition, each scale in a ZNE arm requires a circuit submission and a corresponding result transfer. Thus, the number of round trips and transferred bytes are nondecreasing in the number of scales $m_k$. Therefore, for nested arms with $m_i\leq m_j$ and $S_i\leq S_j$, we also have $R_i\leq R_j$ and $B_i\leq B_j$. This proves the claim.
\end{proof}

\begin{theorem}[Sublinear cost-aware regret of CMAB--ZNE]
\label{thm:cmab_zne_regret}
Suppose Assumptions~\ref{ass:bounded_context_reward}--\ref{ass:cost_aware_utility} hold. If CMAB--ZNE selects
\begin{equation}
    k_t
    =
    \arg\max_{k}
    \left[
    x_t^{\top}\widehat{\theta}_{t,k}
    +
    \beta_T
    \sqrt{x_t^{\top}A_{t,k}^{-1}x_t}
    -
    \lambda_c\widetilde{C}_k
    \right],
    \label{eq:cmab_zne_selection_theory}
\end{equation}
then, with probability at least $1-\delta$, the cumulative cost-aware regret satisfies
\begin{equation}
    R_T
    =
    \mathcal{O}
    \left(
    \beta_T
    \sqrt{
    T K p
    \log
    \left(
    1+\frac{T L^2}{\lambda_0 p}
    \right)}
    \right).
    \label{eq:cmab_zne_regret_bound}
\end{equation}
Consequently, as $T$ increases, the average regret satisfies $R_T/T\rightarrow0$.
\end{theorem}

\begin{proof}
Let $k_t$ be the arm selected by CMAB--ZNE at round $t$, and let $k_t^{\star}$ be the best cost-aware arm. From Lemma~\ref{lem:confidence_bound}, with probability at least $1-\delta$, for every arm $k$,
\begin{equation}
    x_t^{\top}\theta_k^{\star}
    \leq
    x_t^{\top}\widehat{\theta}_{t,k}
    +
    \beta_T
    \sqrt{x_t^{\top}A_{t,k}^{-1}x_t}.
\end{equation}
Because $\lambda_c\widetilde{C}_k$ is known and deterministic for each arm, subtracting the cost term preserves the optimism property:
\begin{equation}
\begin{aligned}
    x_t^{\top}\theta_k^{\star}
    -
    \lambda_c\widetilde{C}_k
    \leq
    x_t^{\top}\widehat{\theta}_{t,k}
    +
    \beta_T
    \sqrt{x_t^{\top}A_{t,k}^{-1}x_t}
    -
    \lambda_c\widetilde{C}_k .
\end{aligned}
\end{equation}
Using this inequality for $k_t^{\star}$, and using the fact that CMAB--ZNE selects the arm with the largest optimistic cost-aware score, we obtain
\begin{equation}
\begin{aligned}
    x_t^{\top}\theta_{k_t^{\star}}^{\star}
    -
    \lambda_c\widetilde{C}_{k_t^{\star}}
    &\leq
    x_t^{\top}\widehat{\theta}_{t,k_t^{\star}}
    +
    \beta_T
    \sqrt{x_t^{\top}A_{t,k_t^{\star}}^{-1}x_t}
    -
    \lambda_c\widetilde{C}_{k_t^{\star}} \\
    &\leq
    x_t^{\top}\widehat{\theta}_{t,k_t}
    +
    \beta_T
    \sqrt{x_t^{\top}A_{t,k_t}^{-1}x_t}
    -
    \lambda_c\widetilde{C}_{k_t}.
\end{aligned}
\end{equation}
Therefore, the instantaneous regret $\Delta_t$ satisfies
\begin{equation}
\begin{aligned}
    \Delta_t
    &=
    \left(
    x_t^{\top}\theta_{k_t^{\star}}^{\star}
    -
    \lambda_c\widetilde{C}_{k_t^{\star}}
    \right)
    -
    \left(
    x_t^{\top}\theta_{k_t}^{\star}
    -
    \lambda_c\widetilde{C}_{k_t}
    \right) \\
    &\leq
    x_t^{\top}
    \left(
    \widehat{\theta}_{t,k_t}
    -
    \theta_{k_t}^{\star}
    \right)
    +
    \beta_T
    \sqrt{x_t^{\top}A_{t,k_t}^{-1}x_t} \\
    &\leq
    2\beta_T
    \sqrt{x_t^{\top}A_{t,k_t}^{-1}x_t},
\end{aligned}
\end{equation}
where the last step again follows from Lemma~\ref{lem:confidence_bound}. Summing over all rounds gives
\begin{equation}
    R_T
    =
    \sum_{t=1}^{T}\Delta_t
    \leq
    2\beta_T
    \sum_{t=1}^{T}
    \sqrt{x_t^{\top}A_{t,k_t}^{-1}x_t}.
\end{equation}
Applying Cauchy--Schwarz,
\begin{equation}
    R_T
    \leq
    2\beta_T
    \sqrt{
    T
    \sum_{t=1}^{T}
    x_t^{\top}A_{t,k_t}^{-1}x_t
    }.
\end{equation}
For each arm, the elliptical potential lemma gives
\begin{equation}
    \sum_{t:k_t=k}
    x_t^{\top}A_{t,k}^{-1}x_t
    =
    \mathcal{O}
    \left(
    p\log
    \left(
    1+\frac{T L^2}{\lambda_0 p}
    \right)
    \right).
\end{equation}
Summing this bound over the $K$ arms yields
\begin{equation}
    \sum_{t=1}^{T}
    x_t^{\top}A_{t,k_t}^{-1}x_t
    =
    \mathcal{O}
    \left(
    Kp\log
    \left(
    1+\frac{T L^2}{\lambda_0 p}
    \right)
    \right).
\end{equation}
Substituting this bound into the preceding inequality proves Eq.~\eqref{eq:cmab_zne_regret_bound}. Since the regret grows sublinearly in $T$, the average regret satisfies $R_T/T\rightarrow0$.
\end{proof}

\begin{remark}
Theorem~\ref{thm:cmab_zne_regret} shows that CMAB--ZNE asymptotically approaches the optimal cost-aware folding-arm selection strategy in hindsight. In the context of this paper, this means that the proposed policy can gradually identify folding arms that balance extrapolation quality and execution overhead, rather than repeatedly using fixed or grid-searched ZNE schedules. The theorem supports the adaptive selection mechanism; it does not claim that ZNE eliminates all higher-order noise effects or that the Lindblad simulator precisely matches physical hardware.
\end{remark}
}

\section{Evaluation}
\label{sec:experiment}

\subsection{Experimental Setup}

\subsubsection{Datasets}

We evaluate on CIFAR-10~\cite{krizhevsky2009learning} and EuroSAT~\cite{helber2019eurosat}. CIFAR-10 is a standard 10-class image-classification benchmark containing \(32\times32\) RGB images, while EuroSAT is a 10-class remote-sensing dataset derived from Sentinel-2 imagery. Together, these datasets allow evaluation on both controlled natural-image classification and real-world domain-specific visual recognition.


\subsubsection{Noise Model}

We implement a customized noise-injected variational quantum circuit in PennyLane to emulate superconducting-device noise under controlled conditions. Noise is inserted both before and after the variational entangling layers to capture cumulative degradation during state preparation, entanglement, and measurement. Specifically, we use depolarizing, bit-flip, phase-flip, amplitude-damping, and phase-damping channels, with a configurable \texttt{noise\_factor} controlling the overall strength. These channels capture the main NISQ error sources: state mixing, computational-basis flips, phase decoherence, energy relaxation, and loss of off-diagonal coherence~\cite{nielsen2010quantum,preskill2018quantum,wang2022quantumnat}.
We analyze gate noise using the MNR model~\cite{fellous2023optimizing}. Following open-system quantum dynamics~\cite{breuer2002theory,wiseman2009quantum} and recent hardware modeling~\cite{geller2020rigorous,fellous2023optimizing}, we use Lindblad-style decoherence to model relaxation and dephasing in superconducting qubits. Classical inputs are encoded through \textit{AngleEmbedding}, followed by noisy layers and \textit{StronglyEntanglingLayers}. The quantum layer outputs Pauli-\(Z\) expectation values, which are passed to the classifier.
\textcolor{black}{Although the main simulation uses controlled PennyLane noise channels, the proposed MNR/Lindblad model can be parameterized by calibration quantities commonly reported by superconducting QPUs, such as $T_1$, $T_2$, gate errors, and readout errors. These quantities allow time-varying rates to be instantiated across calibration windows and enable direct comparison between simulated noisy estimates and real QPU measurements. This connection is further supported by our additional Rigetti hardware validation in Section~\ref{sec:hardware_validation}.}

\textbf{Simulating Shot Noise.}
Shot noise is simulated by sampling analytic Pauli-\(Z\) expectations \(q\in[-1,1]\). Each value is converted to a \(+1\) outcome probability \(p=(1+q)/2\), and counts are sampled as \(k\sim\mathrm{Binomial}(S,p)\) under shot budget \(S\). The empirical estimate is then
\[
\hat q = \frac{2k}{S}-1,
\qquad
\mathbb{E}[\hat q]=q,
\qquad
\mathrm{Var}(\hat q)=\frac{1-q^2}{S}.
\]
Thus, finite-shot evaluation introduces variance that scales as \(O(1/S)\), creating a practical trade-off between estimator stability and hardware/runtime cost.

\subsubsection{Experimental Environment}

\begin{table}[t]
\centering
\footnotesize
\caption{\footnotesize Experimental environment and model configuration.}
\label{tab:env}
\begin{tabular}{ll}
\toprule
\textbf{Aspect} & \textbf{Specification} \\
\midrule
Frameworks & PyTorch, TorchVision, PennyLane, rasterio \\
Hardware & RTX 4090 GPU, 64 GB RAM, Ubuntu 22.04 \\
Batch size / epochs & 64 / 50 \\
Encoders & ResNet-18 RGB encoder; 3-stage CNN MS encoder \\
Fusion & Concat \(1024\) \(\rightarrow\) Linear projection \\
Quantum device & \texttt{default.mixed} \\
Circuit settings & \(n_{\mathrm{wires}}\in\{4,8,16,32\}\), \(n_{\mathrm{layers}}\in\{3,5,7\}\) \\
Noise channels & Depol., bit/phase-flip, amp./phase-damp. \\
Shots & \(S=1024\) \\
Classifier & Linear \(n_{\mathrm{wires}}\rightarrow10\) \\
Optimizer & Adam, learning rate \(10^{-4}\) \\
\bottomrule
\end{tabular}
\vspace{-0.5cm}
\end{table}

All experiments are implemented in Python 3.10 using PyTorch for the classical components and PennyLane for the quantum circuit simulation. TorchVision is used for pretrained visual backbones, and Rasterio is used for Sentinel-2 TIF processing. Table~\ref{tab:env} summarizes the main software, hardware, and architectural settings.

For data handling, CIFAR-10 follows the standard train/test split and is processed as \(32\times32\) RGB input. EuroSAT samples include RGB images and 13-band multispectral (MS) TIF cubes. RGB inputs are resized to \(224\times224\), while MS cubes are scaled by \(10^4\) and bilinearly resampled to the same resolution. The model uses a dual-stream encoder: a ResNet-18 RGB encoder and a lightweight three-stage CNN for the MS cube, each producing a 512-dimensional embedding. The concatenated 1024-dimensional feature is projected to a latent vector compatible with the quantum layer.
{\color{black}
The quantum component uses \texttt{default.mixed}, \texttt{AngleEmbedding}, and \texttt{StronglyEntanglingLayers}. Rather than evaluating only a shallow setting, we test $n_{\mathrm{wires}}\in\{4,8,16,32\}$ and $n_{\mathrm{layers}}\in\{3,5,7\}$ to examine CMAB--ZNE under increasing circuit width and depth. Composite noise is inserted before and after the entangling layers, finite-shot perturbation is applied with $S=1024$, and Pauli-$Z$ expectations from all wires are passed to a linear classifier. This enables comparison of the same hybrid architecture across shallow, medium, and wider VQC settings.
}
The models are trained for 50 epochs using Adam with learning rate \(10^{-4}\). {\color{black} Unless otherwise stated, results are averaged over multiple random seeds and reported with mean and standard deviation to capture variability from stochastic initialization, finite-shot sampling, and noise injection. CMAB--ZNE is compared against fixed-fold and grid-search ZNE under paired runs with the same noise settings; therefore, performance differences primarily reflect the mitigation strategy rather than random trial variation.}

\subsubsection{ZNE Framework}

We apply ZNE at fixed training epochs without retraining the model at each noise level. For each epoch \(i\), the quantum head is evaluated under scaled noise factors \(\lambda\in\Lambda\), the resulting expectation values are fitted with a low-degree polynomial, and the fitted curve is evaluated at \(\lambda=0\) to obtain the noise-free estimate.

The hybrid model contains a variational quantum head using \texttt{AngleEmbedding} and \texttt{StronglyEntanglingLayers}, followed by a linear classifier. At each epoch, we save the quantum head parameters and classifier weights. During offline ZNE, inference is performed using \texttt{default.mixed} with scaled channel strength \(p(\lambda)=\lambda p_{\mathrm{base}}\), where the base noise includes bit/phase flips, depolarization, amplitude damping, and phase damping. For each epoch \(i\) and scale \(\lambda\in\{0,1,3,5,7\}\), we store the analytic expectation vector \(q^{(\lambda)}\in[-1,1]^{n_{\mathrm{wires}}}\), together with the filename and label. These dumps are the only inputs used for offline ZNE reconstruction.

For a selected arm \(\Lambda\), such as \(\{1,3\}\), \(\{1,3,5\}\), or \(\{1,3,5,7\}\), we perform wire-wise polynomial regression over \(\lambda\) and evaluate the fitted polynomial at zero:
\begin{equation}
    \hat{q}_j(0)=p_j(0), \qquad
    p_j=\arg\min_{\deg(p)\le d}
    \sum_{\lambda\in\Lambda}
    \left(p(\lambda)-q_j^{(\lambda)}\right)^2 ,
    \label{eq:offline_zne_fit}
\end{equation}
where \(j\) indexes the measured wire and \(d=\min(|\Lambda|-1,2)\). The extrapolated value is clipped to \([-1,1]\) for numerical stability, producing \(\hat{q}(0)\in\mathbb{R}^{n_{\mathrm{wires}}}\) for each sample. This reconstructed zero-noise expectation vector is then passed to the frozen classifier for prediction and reward computation.

\subsubsection{CMAB Settings}

We implement the CMAB framework from Subsection~\ref{subsec:zne-cmab} using three ZNE arms,
\[
\mathcal{A}=\{\{1,3\},\{1,3,5\},\{1,3,5,7\}\}.
\]
At each epoch \(t\), the agent selects an arm \(a_t\in\mathcal{A}\), applies it to the stored expectation-value dumps, and receives a scalar reward.
The context vector \(x_t\) is constructed from quantities known before arm selection, including the epoch index, loss and accuracy at \(\lambda=1\), variance across the \(\{1,3\}\) scales, and batch size. Therefore, context construction does not require additional circuit executions. The reward follows Eq.~\eqref{eq:reward_detailed_cmab_zne}, with target depth \(d_0=1\) and \(\alpha\in[10^{-2},10^{-1}]\), using signals such as true-class probability, per-example loss, or average extrapolated expectations. This encourages stable and shallow arms unless deeper extrapolation provides a clear gain.
{\color{black}
For offline evaluation, the simulation saves expectation-value dumps at every epoch for all tested scales $\lambda\in\{1,3,5,7\}$. Thus, even when the bandit selects only one arm $a_t$, we can reconstruct the ZNE output and compute rewards for all unselected arms under the same context. We refer to these as \textit{counterfactual rewards}. They are used only for offline replay, oracle comparison, regret computation, arm-selection analysis, and cumulative-reward reporting. In real hardware or online deployment, only the reward of the selected arm is observed.
}

\subsection{Simulation Results}
\textbf{Ablation Study.} 
{\color{black}
We first present the ablation summary to assess how circuit width and depth affect model performance. Unlike the previous shallow-only setting, we include configurations with up to 32 wires and 7 ansatz layers to better analyze the proposed approach under stronger noise accumulation. Increasing the number of wires expands the dimension of the quantum representation, while increasing the number of layers increases the number of parameterized rotations and entangling operations. The results show that wider circuits can improve expressivity, but excessive depth can also increase noise sensitivity and finite-shot variance. Therefore, the selected configuration reflects a balance among representation capacity, validation loss, and noise-induced overhead.}
\begin{table*}[t]
\color{black}
  \centering
  \caption{\footnotesize Ablation on the number of wires ($n_{\mathrm{wires}}$) and layers.
  We report the final validation loss as mean $\pm$ standard deviation over 3 seeds (lower is better). The best setting used in the main experiments is in bold.}
  \label{tab:ablation_wires_layers}
  \setlength{\tabcolsep}{4pt}
  \renewcommand{\arraystretch}{1.12}
  \begin{tabular}{c ccc ccc}
    \toprule
    & \multicolumn{3}{c}{CIFAR-10 loss} & \multicolumn{3}{c}{EuroSAT loss} \\
    \cmidrule(lr){2-4} \cmidrule(lr){5-7}
    $n_{\mathrm{wires}}$ 
    & $n_{\mathrm{layers}}{=}3$ 
    & $5$ 
    & $7$ 
    & $3$ 
    & $5$ 
    & $7$ \\
    \midrule
    4  
      & \(0.0675 \pm 0.0027\) 
      & \(0.0712 \pm 0.0029\) 
      & \(0.0784 \pm 0.0033\)
      & \(0.1569 \pm 0.0042\) 
      & \(0.1621 \pm 0.0045\) 
      & \(0.1716 \pm 0.0048\) \\

    8  
      & \(0.0288 \pm 0.0016\) 
      & \(0.0362 \pm 0.0019\) 
      & \(0.0437 \pm 0.0024\)
      & \(0.1120 \pm 0.0031\) 
      & \(0.1205 \pm 0.0034\) 
      & \(0.1298 \pm 0.0038\) \\

    16 
    & \(0.0264 \pm 0.0017\) 
      & \(0.0242 \pm 0.0015\) 
      & \(0.0318 \pm 0.0021\)
      & \(0.1058 \pm 0.0030\) 
      & \(0.1013 \pm 0.0028\) 
      & \(0.1127 \pm 0.0035\) \\

    32 
      & \(0.0248 \pm 0.0014\) 
      & \(\mathbf{0.0219 \pm 0.0012}\) 
      & \(0.0276 \pm 0.0016\)
      & \(0.1032 \pm 0.0028\) 
      & \(\mathbf{0.0978 \pm 0.0025}\) 
      & \(0.1065 \pm 0.0030\) \\
    \bottomrule
  \end{tabular}
\end{table*}
\begin{figure}[t]
  \centering

  \begin{subfigure}[t]{0.24\textwidth}
    \centering
    \includegraphics[width=\linewidth]{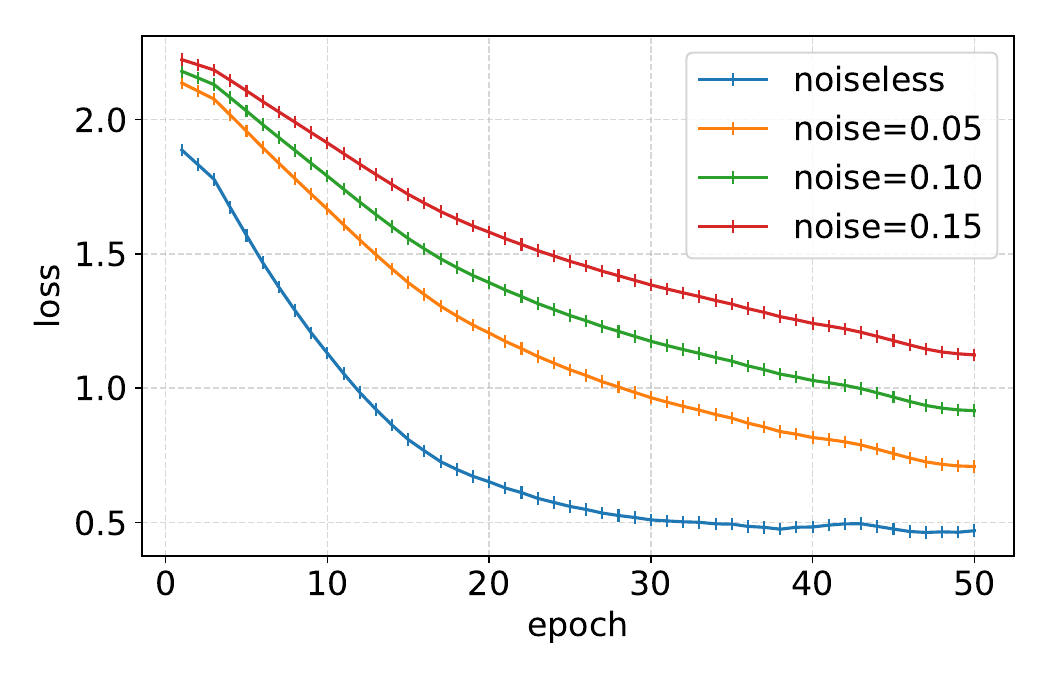}
    \vspace{-2em}
    \caption{\footnotesize CIFAR-10: loss vs. epoch}
    \label{fig:cifar-loss}
  \end{subfigure}
  \hfill
  \begin{subfigure}[t]{0.24\textwidth}
    \centering
    \includegraphics[width=\linewidth]{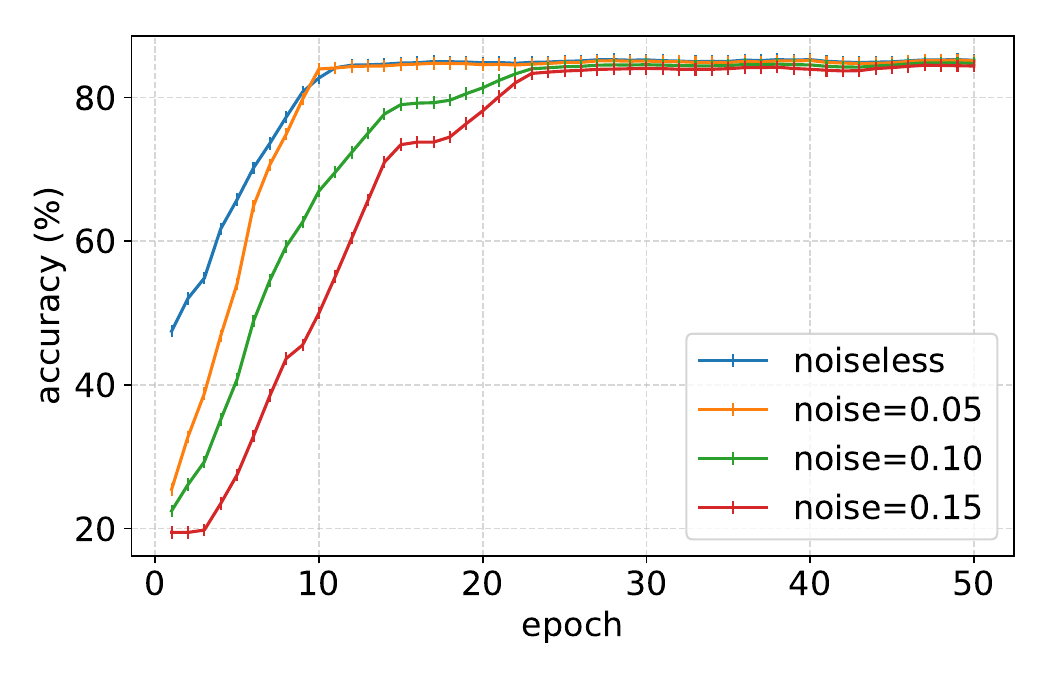}
    \vspace{-2em}
    \caption{\footnotesize CIFAR-10: accuracy vs. epoch}
    \label{fig:cifar-acc}
  \end{subfigure}

  \vspace{0.75em}

  \begin{subfigure}[t]{0.24\textwidth}
    \centering
    \includegraphics[width=\linewidth]{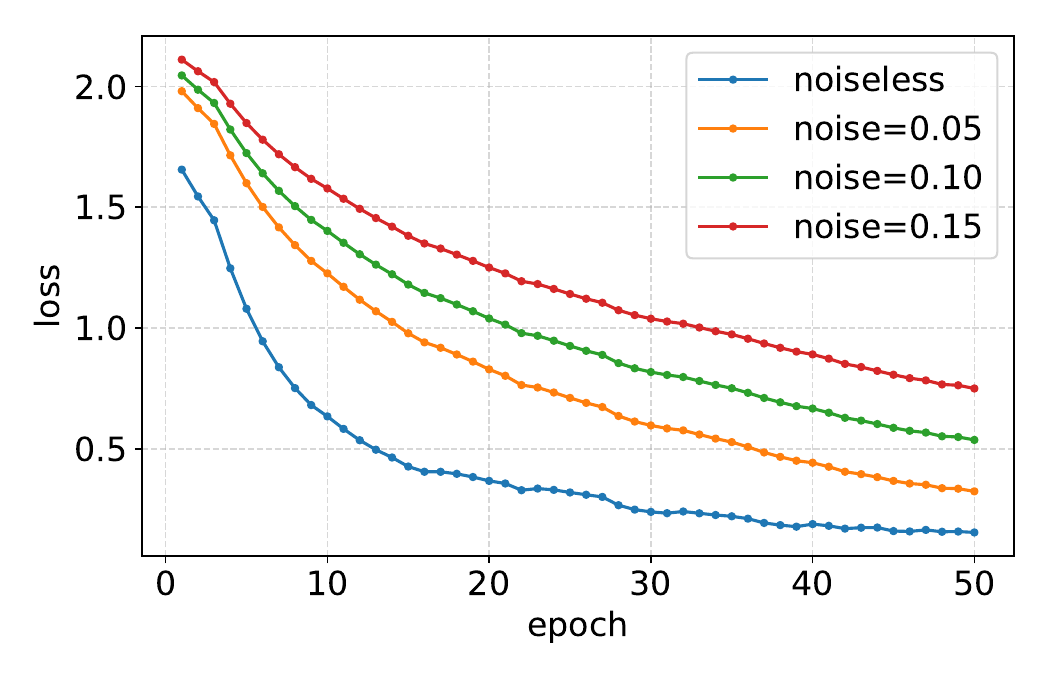}
    \vspace{-2em}
    \caption{\footnotesize EuroSAT: loss vs. epoch}
    \label{fig:eurosat-loss}
  \end{subfigure}
  \hfill
  \begin{subfigure}[t]{0.24\textwidth}
    \centering
    \includegraphics[width=\linewidth]{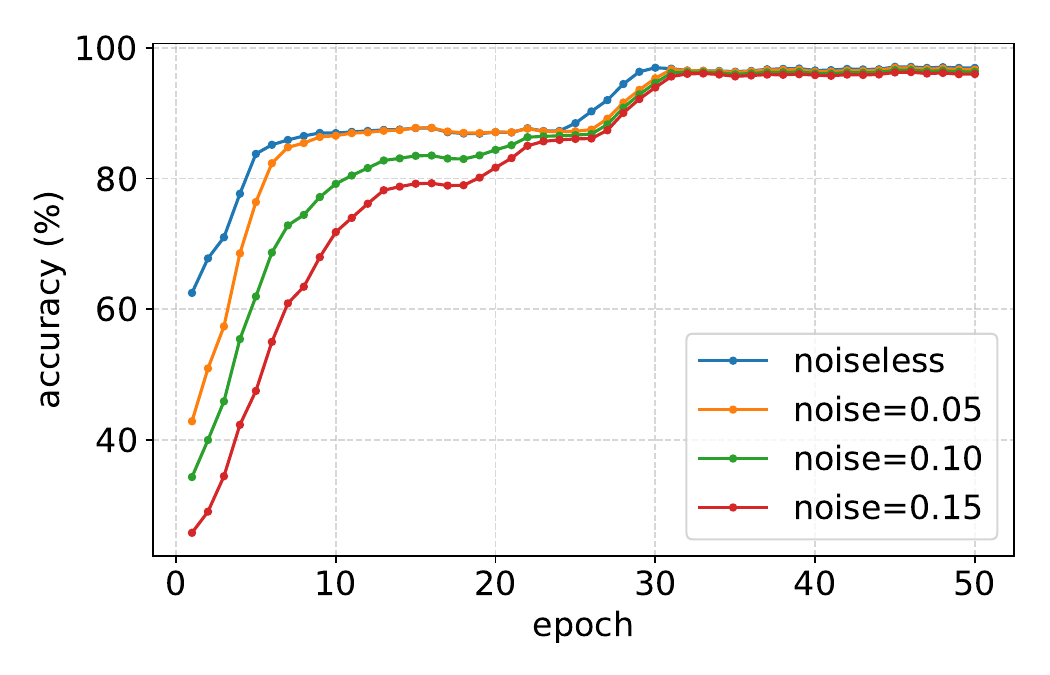}
    \vspace{-2em}
    \caption{\footnotesize EuroSAT: accuracy vs. epoch}
    \label{fig:eurosat-acc}
  \end{subfigure}

  \caption{\footnotesize Noiseless vs. noise levels (0.05, 0.10, 0.15). Each curve is smoothed with the same window as in preprocessing.}
  \label{fig:noise}
\end{figure}

\textbf{Impact of Quantum Noise.}
Fig.~\ref{fig:noise} shows training loss and accuracy against epoch for two datasets (CIFAR-10, EuroSAT) with four settings: \emph{noiseless} and noise levels $0.05$, $0.10$, and $0.15$. The $0.10$ curve is optionally imputed point-wise as the mean of the $0.05$ and $0.15$ curves.  All curves have been smoothed using the same window as in preprocessing.  Across all datasets, increasing noise consistently slows optimization—loss declines more gradually, and accuracy grows more slowly—resulting in the greatest gap across curves in the early to mid epochs.  The noiseless alternative provides the quickest loss reduction and accuracy increase, followed by $0.05$, $0.10$, and $0.15$.
Despite the initial lag caused by noise, the trajectories tend to converge to comparable terminal behavior in subsequent epochs, with only minor residual gaps between noise levels.  Throughout training, the imputed $0.10$ curve remained between the $0.05$ and $0.15$ trajectories, indicating a near-linear influence of moderate noise in this range.  Overall, the figure shows that moderate noise mostly delays convergence rather than fundamentally restricting the possible performance in our configuration, with EuroSAT showing somewhat higher asymptotic accuracy and CIFAR-10 indicating more early-epoch susceptibility to noise. 

\textbf{Effect of Measurement Shot Count.}
\begin{figure}[t]
  \centering

  \begin{subfigure}[t]{0.24\textwidth}
    \centering
    \includegraphics[width=\linewidth]{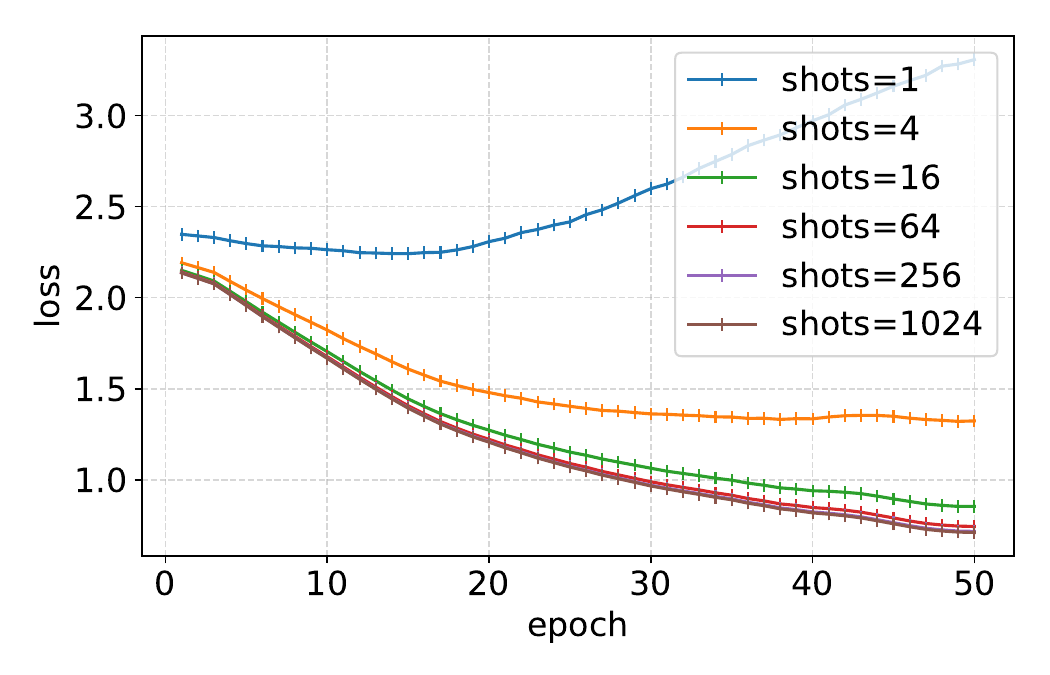}
    \vspace{-2em}
    \caption{\footnotesize CIFAR-10: loss comparison for different shots}
    \label{fig:cifar-loss-2}
  \end{subfigure}
  \hfill
  \begin{subfigure}[t]{0.24\textwidth}
    \centering
    \includegraphics[width=\linewidth]{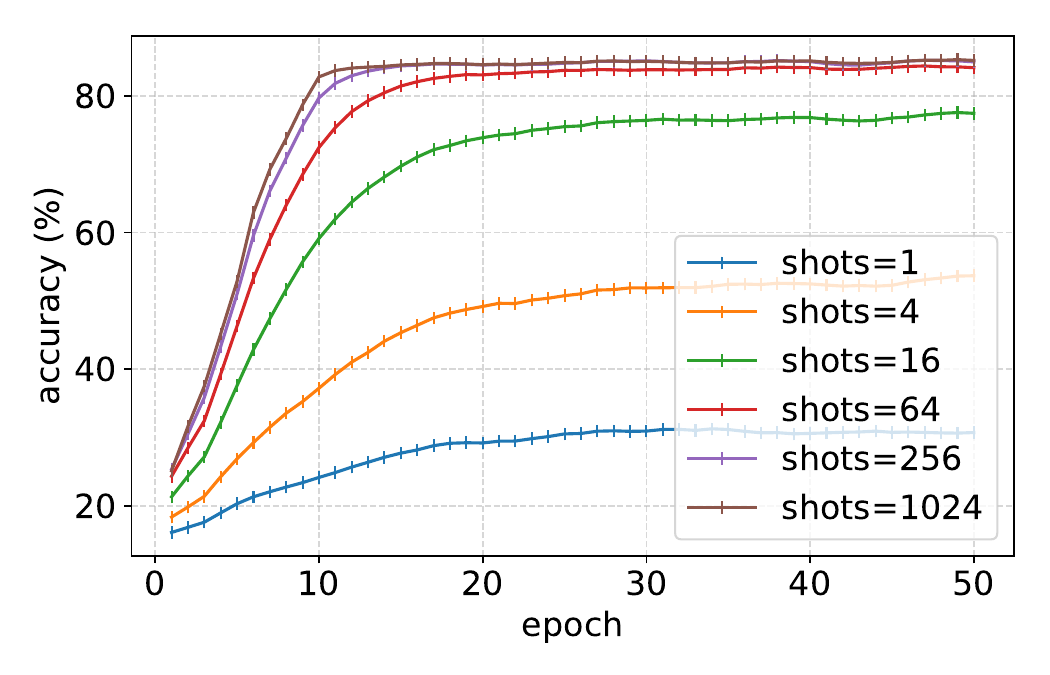}
    \vspace{-2em}
    \caption{\footnotesize CIFAR-10: accuracy comparison for different shots}
    \label{fig:cifar-acc-2}
  \end{subfigure}

  \vspace{0.75em}

  \begin{subfigure}[t]{0.24\textwidth}
    \centering
    \includegraphics[width=\linewidth]{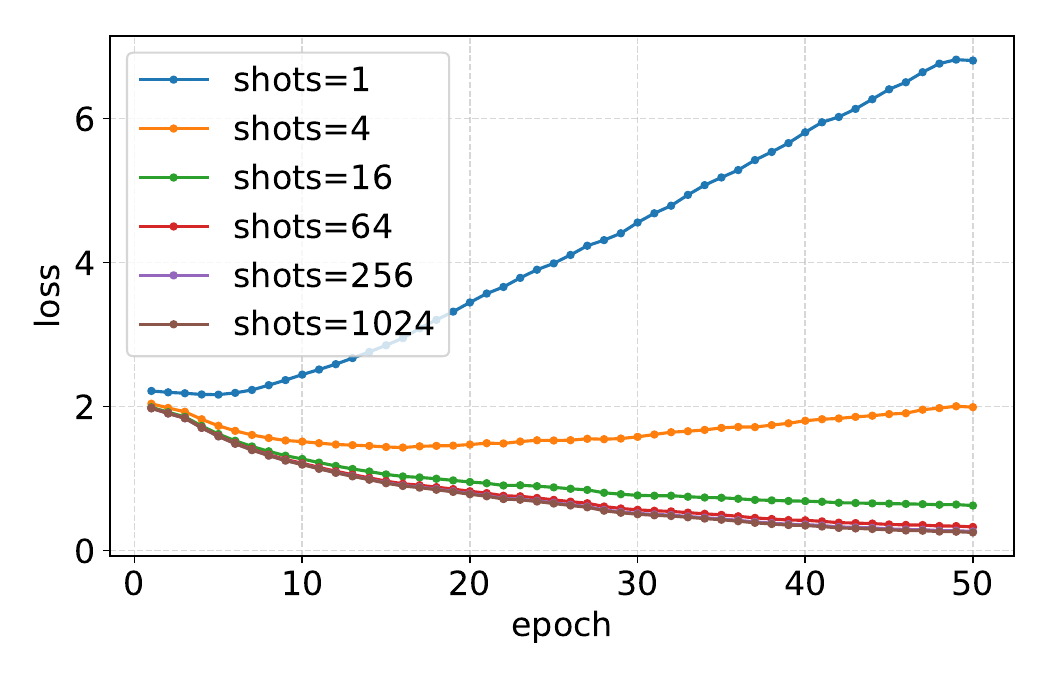}
    \vspace{-2em}
    \caption{\footnotesize EuroSAT: loss comparison for different shots}
    \label{fig:eurosat-loss-2}
  \end{subfigure}
  \hfill
  \begin{subfigure}[t]{0.24\textwidth}
    \centering
    \includegraphics[width=\linewidth]{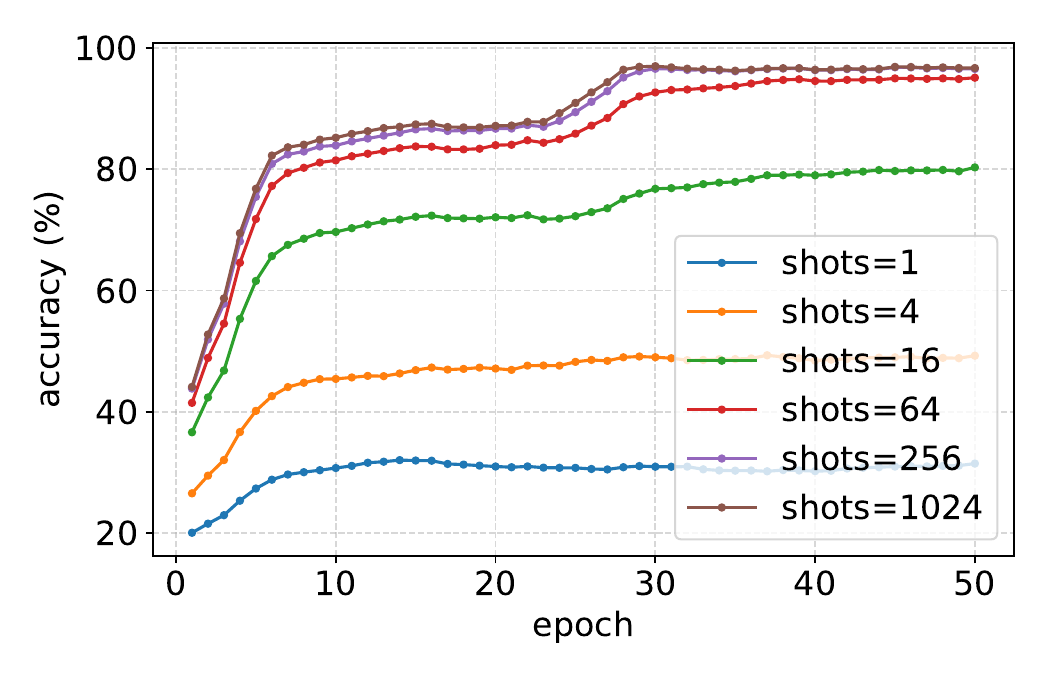}
    \vspace{-2em}
    \caption{\footnotesize EuroSAT: accuracy comparison for different shots}
    \label{fig:eurosat-acc-2}
  \end{subfigure}

  \caption{\footnotesize Effect of finite sampling (shots) on model performance at a fixed hardware noise level (0.05). Curves show shots $s =  \{1, 4, 16, 64, 256, 1024\}$; higher shot counts reduce measurement variance and yield smoother trajectories.}
  \label{fig:shots}
\end{figure}
Fig.~\ref{fig:shots} shows training loss and accuracy vs epoch on CIFAR-10 and EuroSAT with a fixed hardware noise level ($\epsilon=0.05$) and altering the number of measurement shots $s\in{1,4,16,64,256,1024}$. Across both datasets, greater $s$ results in smoother trajectories and quicker optimization: loss lowers faster, and accuracy increases faster, but very few shots (e.g., $s=1$ or $4$) show significant variation, slower progress, and, in EuroSAT, even unstable loss behavior.
As $s$ grows, performance benefits decline (curves for $s!\ge!256$ are near), but residual advantages remain noticeable, especially in early epochs and with noisier gradient estimations. The overall pattern shows that higher sampling decreases measurement variation, stabilizes training, and increases asymptotic accuracy. \emph{Overall, the figure demonstrates that employing more shots improves performance}.
\begin{figure*}
  \centering

  \begin{subfigure}[t]{0.32\textwidth}
    \centering
    \includegraphics[width=\linewidth]{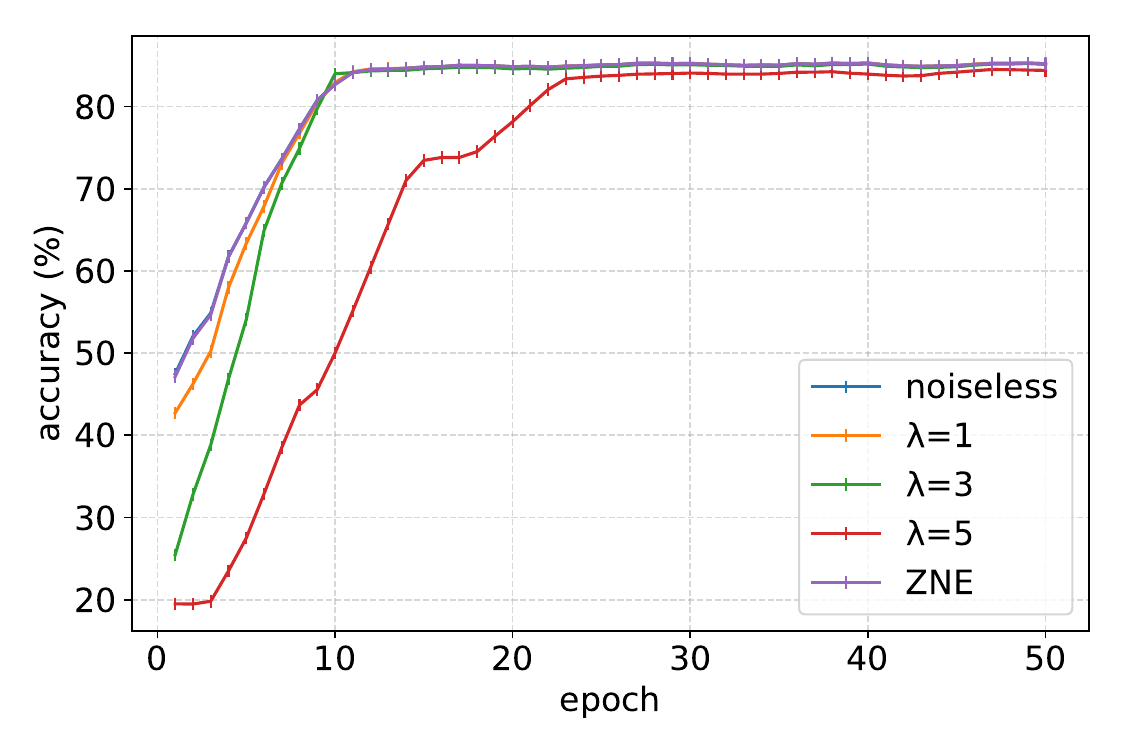}
    \subcaption{\footnotesize CIFAR-10 Accuracy}
  \end{subfigure}
  \hfill
  \begin{subfigure}[t]{0.32\textwidth}
    \centering
    \includegraphics[width=\linewidth]{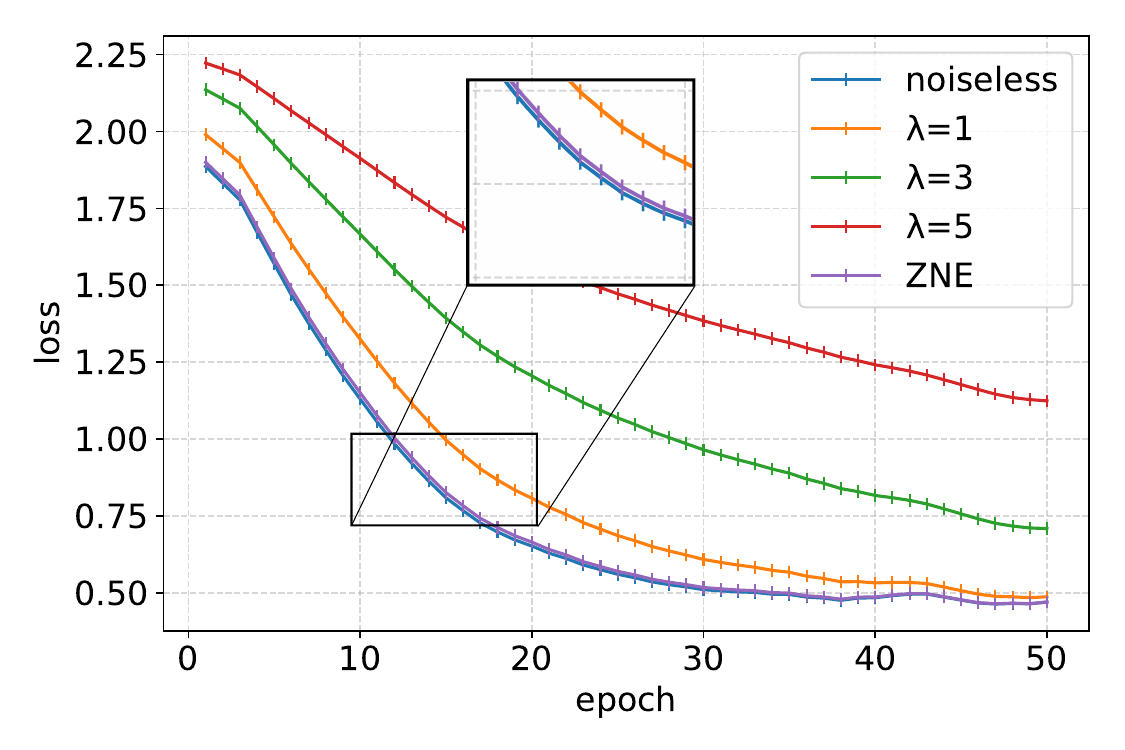}
    \subcaption{\footnotesize CIFAR-10 Loss}
  \end{subfigure}
  \hfill
  \begin{subfigure}[t]{0.32\textwidth}
    \centering
    \includegraphics[width=\linewidth]{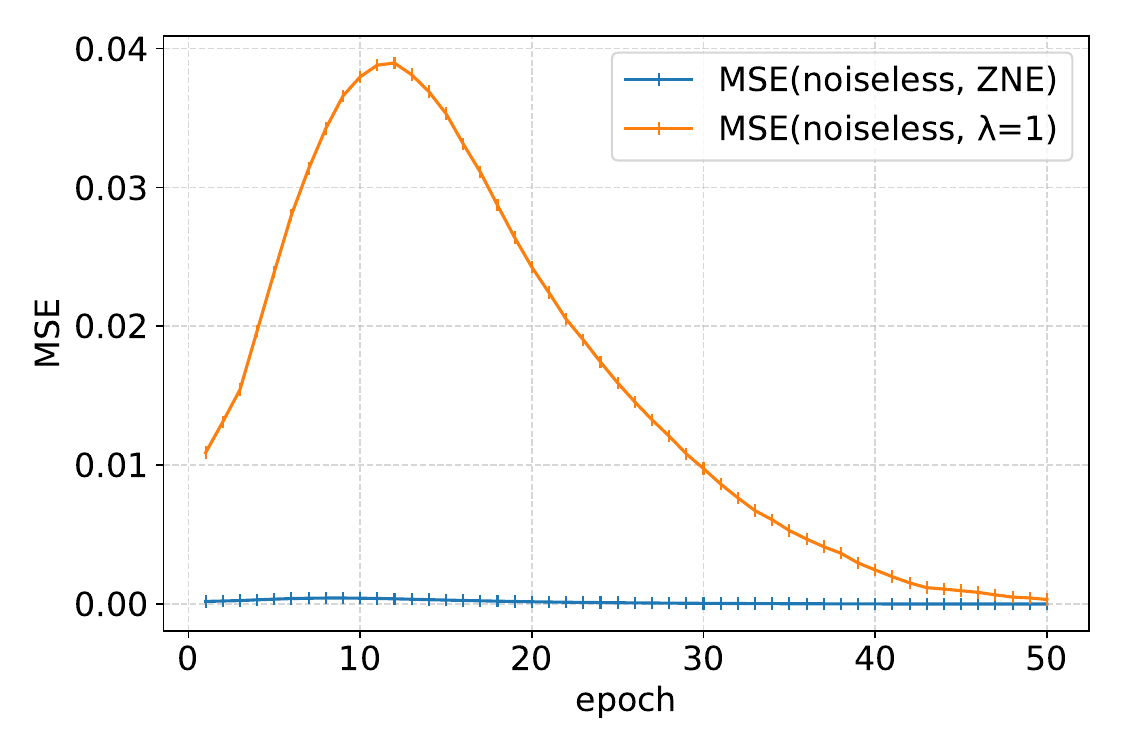}
    \subcaption{\footnotesize CIFAR-10 MSE}
  \end{subfigure}

  \vspace{0.8em}

  \begin{subfigure}[t]{0.32\textwidth}
    \centering
    \includegraphics[width=\linewidth]{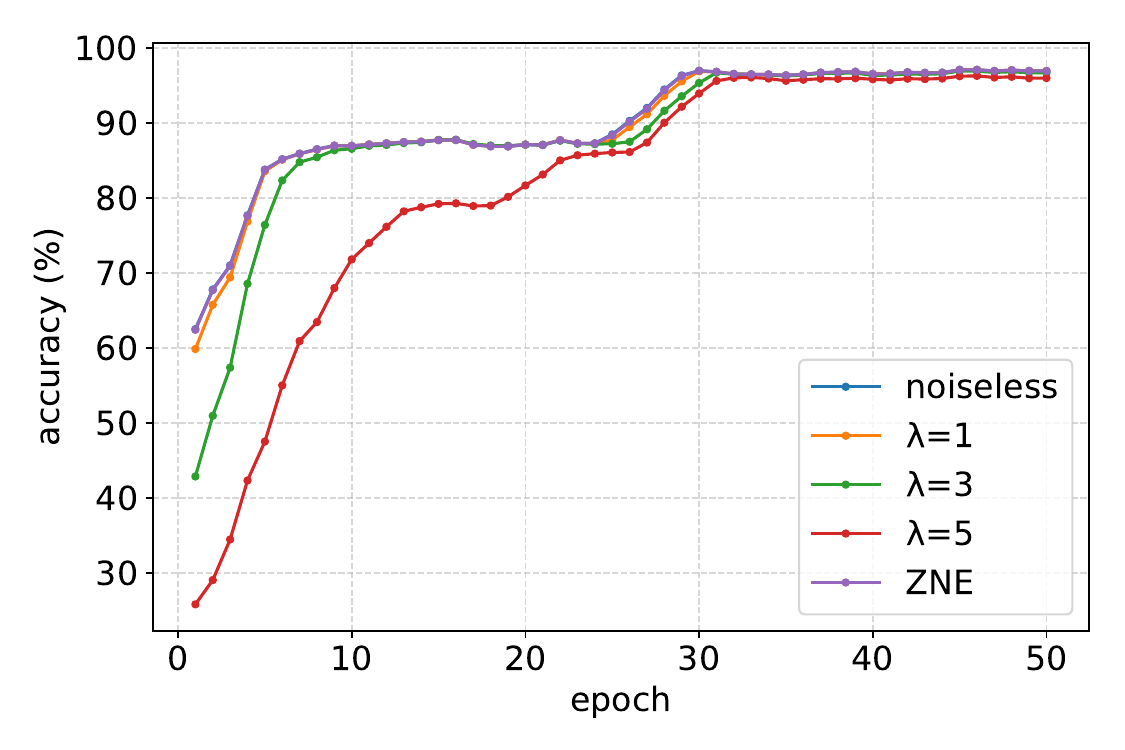}
    \subcaption{\footnotesize EuroSAT Accuracy}
  \end{subfigure}
  \hfill
  \begin{subfigure}[t]{0.32\textwidth}
    \centering
    \includegraphics[width=\linewidth]{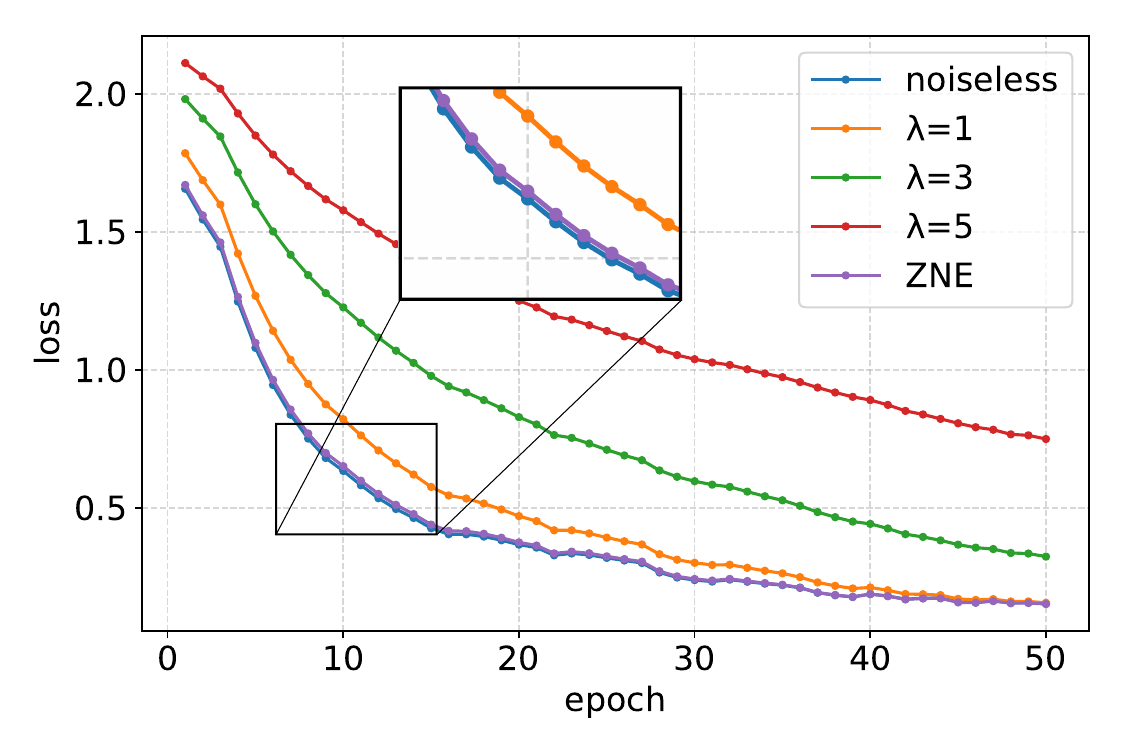}
    \subcaption{\footnotesize EuroSAT Loss}
  \end{subfigure}
  \hfill
  \begin{subfigure}[t]{0.32\textwidth}
    \centering
    \includegraphics[width=\linewidth]{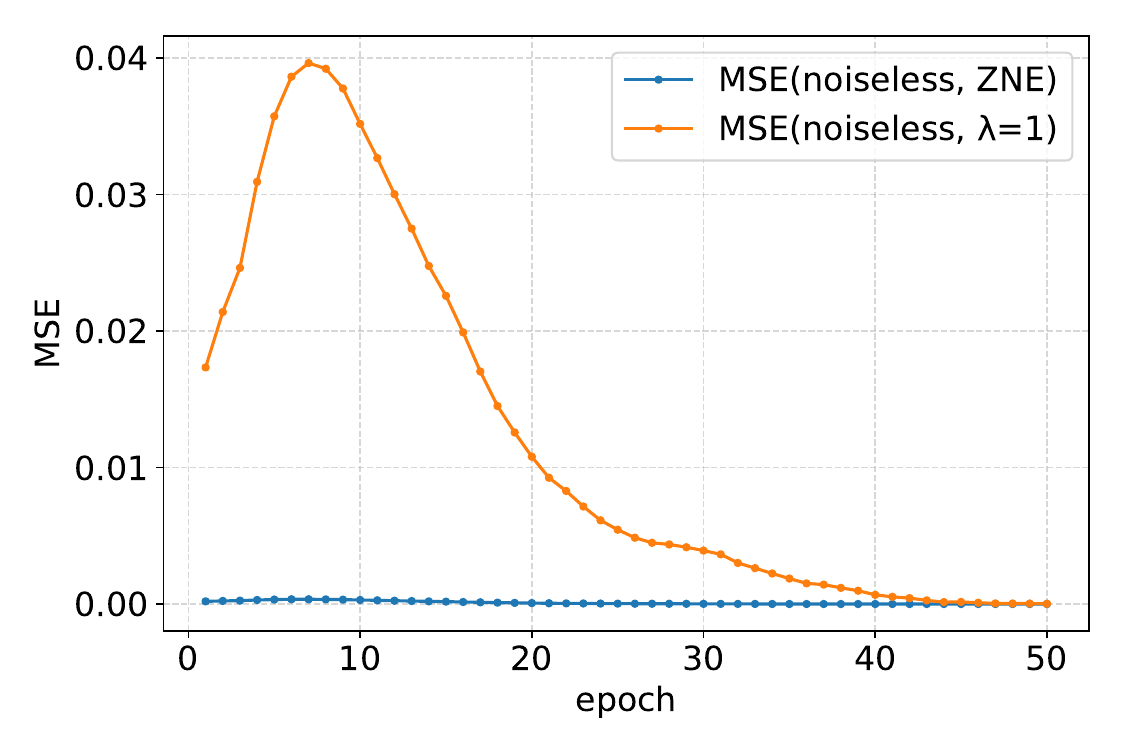}
    \subcaption{\footnotesize EuroSAT MSE}
  \end{subfigure}

  \caption{\footnotesize Comparing noiseless ($\lambda=0$), noisy ($\lambda\in\{1,3,5\}$), and zero-noise extrapolation (ZNE) used in quantum neural networks where $n_{\mathrm{wires}}{=}16$ and $n_{\mathrm{layers}}{=}5$.  ZNE continuously outperforms the noisy runs and closely follows the noiseless baseline in terms of accuracy and loss.  ZNE's recovery of noise-free behavior is illustrated by the MSE panel, which demonstrates that $\mathrm{MSE}(\text{noiseless},\text{ZNE})$ stays close to zero throughout epochs, while $\mathrm{MSE}(\text{noiseless},\lambda=1)$ is significantly higher early and decreases with training.}
  \label{fig:zne_cifar_eurosat}
\end{figure*}

\textbf{ZNE results.} Training dynamics are compared with and without zero-noise extrapolation (ZNE) in Fig.~\ref{fig:zne_cifar_eurosat}. From left to right, accuracy, loss, and the mean-squared error (MSE) for each epoch in relation to the noiseless baseline are displayed in each row (top: CIFAR-10; bottom: EuroSAT). Curves contain the ZNE estimate extrapolated from the noisy runs, the noiseless run ($\lambda{=}0$), and the noisy executions at $\lambda\in\{1,3,5\}$. As expected, greater noise scales increase loss and considerably delay the accuracy ramp-up in early epochs; by contrast, ZNE closely matches the noiseless trajectory during the entire training, producing virtually identical accuracy and very small loss variances. 

These gaps are measured by the panels on the right $\mathrm{MSE}(\text{noiseless},\lambda{=}1)$ while $\mathrm{MSE}(\text{noiseless},\text{ZNE})$ stays close to zero throughout—usually at least an order of magnitude smaller than the noisy baselines—is highest at the beginning of training and decreases as the model converges.  The modest residuals are in line with higher-order noise effects and typical training stochasticity, demonstrating that ZNE successfully eliminates the primary noise-induced bias while maintaining the optimization dynamics of the noiseless model.  There is comparable behavior on both datasets.

\textbf{ZNE Across Multiple Arms.} {\color{black}
We evaluate every candidate arm in
\[
\mathcal{A}_4=\{[1,3],[1,3,5],[1,3,5,7],[1,3,5,7,9,11,13]\}
\]
to further examine the impact of action-space expansion. Each arm represents a distinct ZNE folding-scale set. Although larger arms provide more extrapolation points, they also require deeper circuits, more folded-circuit executions, higher shot usage, and greater communication overhead. Thus, this experiment examines whether the accuracy gain from larger arms justifies the additional end-to-end cost.
}
{\color{black}
Table~\ref{tab:all-arm-comparison} shows that increasing the number of folding scales improves accuracy and loss, but the improvement diminishes as the arm size increases. For CIFAR-10, moving from \([1,3]\) to \([1,3,5,7]\) increases accuracy from \(84.96\%\) to \(85.58\%\), while the total cost increases by \(19.7\%\). Extending the arm to \([1,3,5,7,9,11,13]\) further increases accuracy only to \(85.91\%\), while the cost rises by \(38.4\%\). A similar trend appears on EuroSAT, where accuracy increases from \(96.02\%\) to \(96.18\%\) when comparing \([1,3,5,7]\) with the largest arm, but the end-to-end cost increases from \(20.9\%\) to \(41.7\%\). These results show that larger folding arms provide greater extrapolation flexibility, but their marginal accuracy gain is small relative to the additional execution and communication overhead. This validates the use of CMAB--ZNE as a cost-aware policy that determines when expensive arms are worth selecting under the current noise context.
}

\begin{table}[t]
\centering
\scriptsize
\color{black}
\caption{\footnotesize Comparison of all CMAB candidate arms under dynamic noise $(\widehat{\eta}_t)$. Accuracy and loss are averaged across epochs and reported as mean $\pm$ standard deviation over 3 seeds at circuit depth 3 and noise band $\eta=0.05$ with a 10 Mbps link budget.}
\label{tab:all-arm-comparison}
\setlength{\tabcolsep}{3pt}
\renewcommand{\arraystretch}{1.12}
\begin{tabular}{lcccccc}
\toprule
Dataset & Arm \(\Lambda_k\) & \# scales & \(S_k\) & Acc. (\%) & Loss\\
\midrule
\multirow{4}{*}{CIFAR-10}
& \([1,3]\) 
& 2 & 4 
& \(84.96 \pm 0.45\) 
& \(0.437 \pm 0.012\) 
\\

& \([1,3,5]\) 
& 3 & 9 
& \(85.34 \pm 0.43\) 
& \(0.428 \pm 0.011\) 
\\

& \([1,3,5,7]\) 
& 4 & 16 
& \(85.58 \pm 0.42\) 
& \(0.423 \pm 0.011\) 
\\

& \([1,3,5,7,9,11,13]\) 
& 7 & 49 
& \(85.91 \pm 0.39\) 
& \(0.417 \pm 0.010\) 
\\

\midrule
\multirow{4}{*}{EuroSAT}
& \([1,3]\) 
& 2 & 4 
& \(95.61 \pm 0.34\) 
& \(0.166 \pm 0.009\) 
\\
 
& \([1,3,5]\) 
& 3 & 9 
& \(95.87 \pm 0.32\) 
& \(0.161 \pm 0.008\) 
 \\

& \([1,3,5,7]\) 
& 4 & 16 
& \(96.02 \pm 0.31\) 
& \(0.157 \pm 0.008\) 
 \\
 
& \([1,3,5,7,9,11,13]\) 
& 7 & 49 
& \(96.18 \pm 0.29\) 
& \(0.154 \pm 0.007\) 
 \\
\bottomrule
\end{tabular}
\end{table}

Fig.~\ref{fig:delta-grid} reports the epoch-wise ZNE gap relative to the noiseless baseline for accuracy and loss on CIFAR-10 and EuroSAT using three arm sets: \(\{1,3\}\), \(\{1,3,5\}\), and \(\{1,3,5,7\}\). Positive \(\Delta\) indicates that ZNE is above the noiseless baseline, while negative \(\Delta\) indicates that it is below. Overall, larger arm sets reduce the ZNE bias and move the estimates closer to the noiseless trajectory, especially in later epochs. On CIFAR-10, the initial accuracy/loss gap decreases rapidly with training; on EuroSAT, \(\{1,3,5,7\}\) gives the most stable improvement.

\begin{figure}
  \centering
  \begin{subfigure}{0.24\textwidth}
    \includegraphics[width=\linewidth]{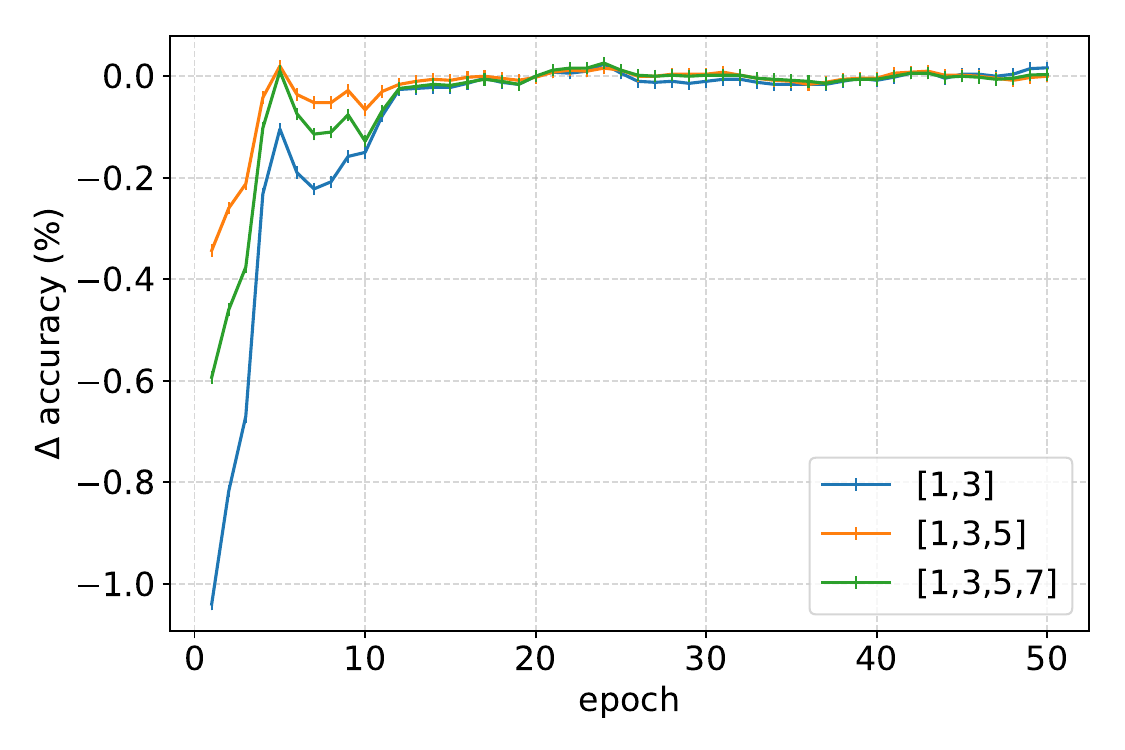}
    \caption{\footnotesize CIFAR-10: $\Delta$ accuracy (ZNE $-$ noiseless)}
    \label{fig:cifar-delta-acc}
  \end{subfigure}\hfill
  \begin{subfigure}{0.24\textwidth}
    \includegraphics[width=\linewidth]{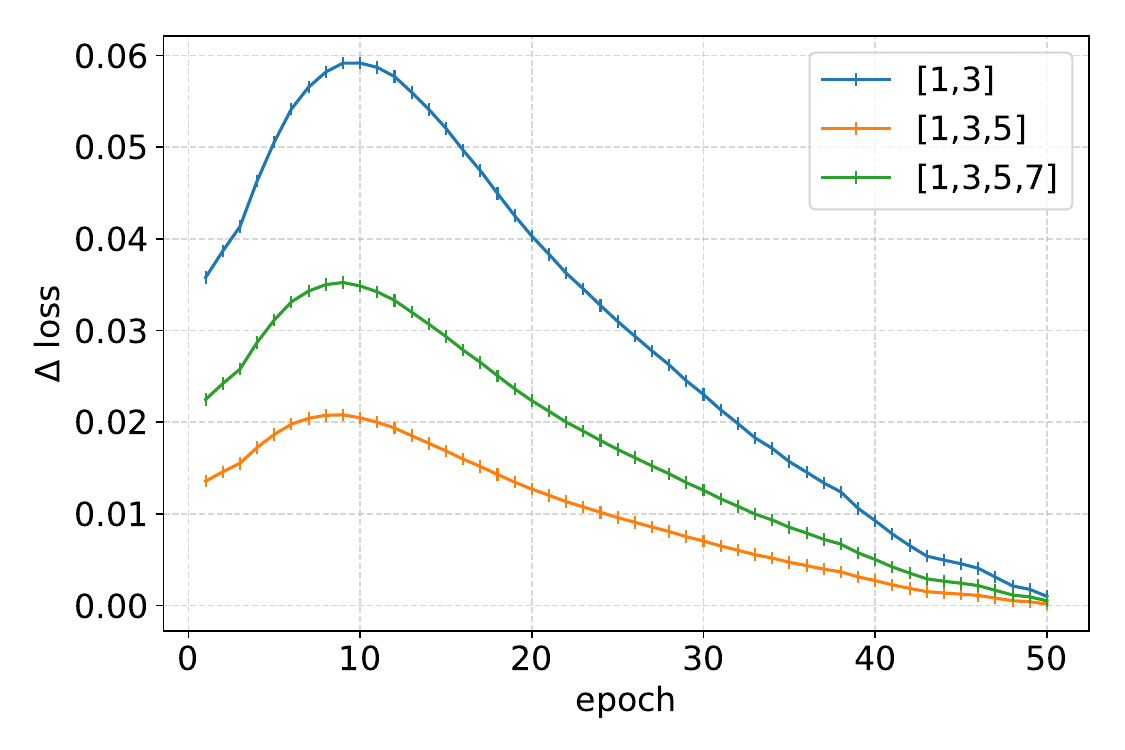}
    \caption{\footnotesize CIFAR-10: $\Delta$ loss (ZNE $-$ noiseless)}
    \label{fig:cifar-delta-loss}
  \end{subfigure}

  \vspace{0.5em}

  \begin{subfigure}{0.24\textwidth}
    \includegraphics[width=\linewidth]{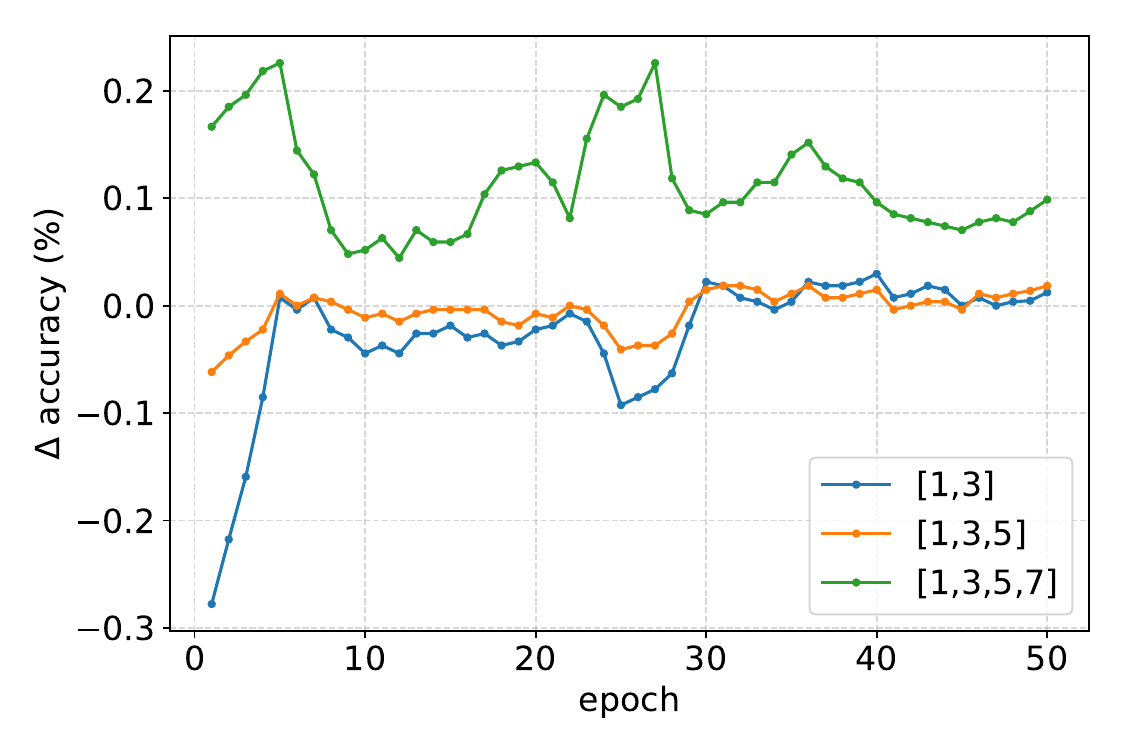}
    \caption{\footnotesize EuroSAT: $\Delta$ accuracy (ZNE $-$ noiseless)}
    \label{fig:euro-delta-acc}
  \end{subfigure}\hfill
  \begin{subfigure}{0.24\textwidth}
    \includegraphics[width=\linewidth]{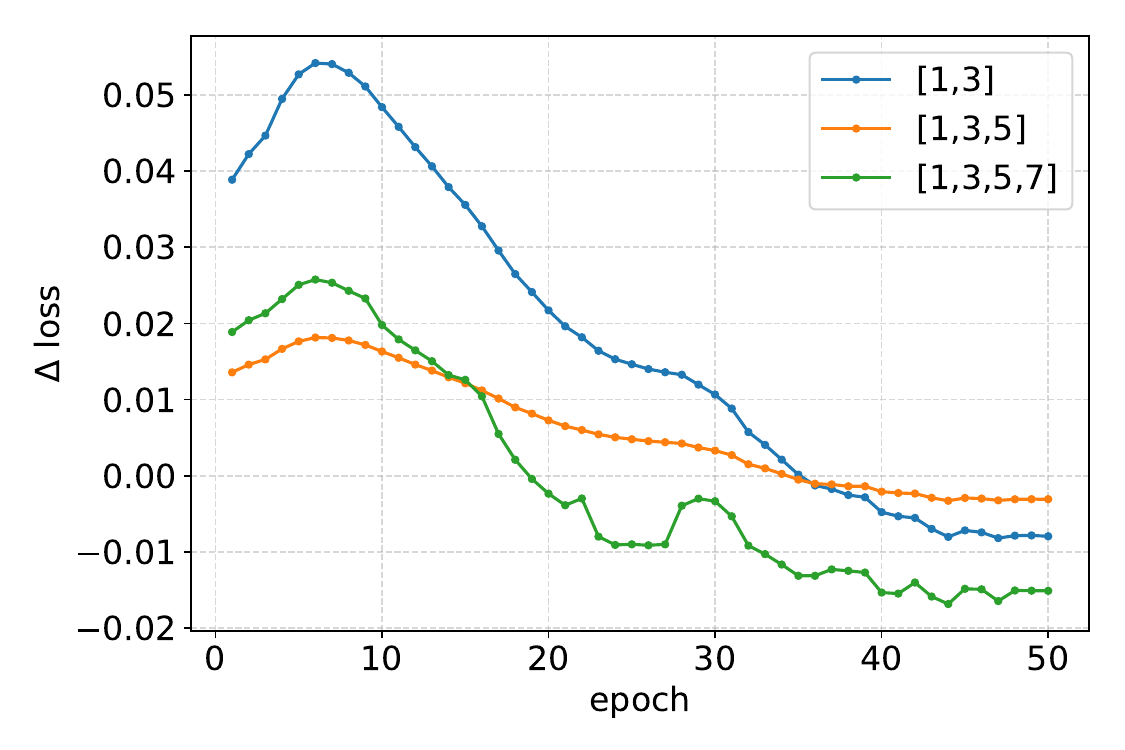}
    \caption{\footnotesize EuroSAT: $\Delta$ loss (ZNE $-$ noiseless)}
    \label{fig:euro-delta-loss}
  \end{subfigure}

  \caption{\footnotesize Zero-noise extrapolation (ZNE) deltas for different arm sets
  $\{1,3\}$, $\{1,3,5\}$, and $\{1,3,5,7\}$. Curves show the difference
  between the ZNE estimate and the noiseless baseline at each epoch for
  accuracy (left) and loss (right).}
  \label{fig:delta-grid}
\end{figure}

\textbf{Simulation on Dynamic Noise.} To show the effectiveness of the proposed CMAB-ZNE approach, we first add a dynamic noise profile instead of a static noise as illustrated in Fig.~\ref{fig:noise-proxy}. The normalized noise proxy $\hat{\eta}_t$ shows non-stationary (time-varying) hardware noise throughout training, increasing slowly from $\approx0.25$ at early epochs to $\approx1.0$ by $\sim$40 epochs before saturating with minor oscillations.  Since the effective noise level increases with training, this tendency encourages adaptive mitigation and sampling approaches (instead of fixed schedules).
\begin{figure}[t]
  \centering
  \includegraphics[width=0.69\linewidth]{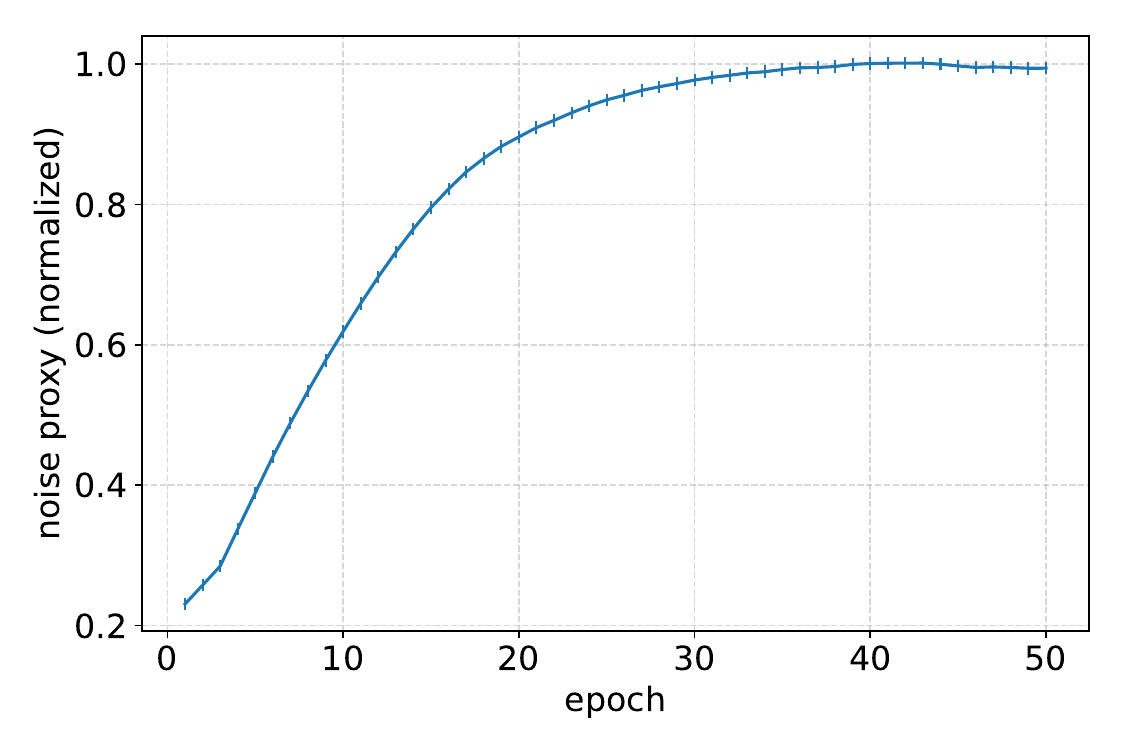}
  \caption{\footnotesize Normalized noise proxy $\hat{\eta}_t$ versus training epoch.}
  \vspace{-1em}
  \label{fig:noise-proxy}
\end{figure}

\begin{figure}
  \centering
  \begin{subfigure}[t]{0.24\textwidth}
    \centering
    \includegraphics[width=\linewidth]{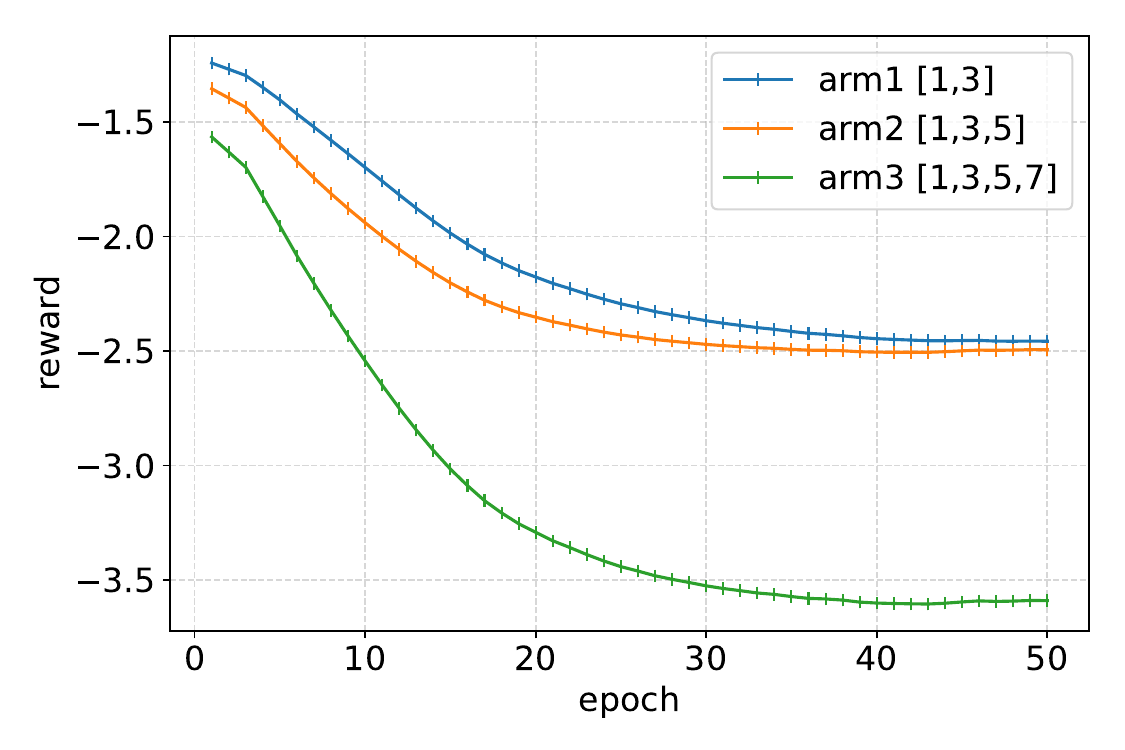}
    \caption{\footnotesize CIFAR-10: reward per arm}
    \label{fig:cifar-reward2}
  \end{subfigure}\hfill
  \begin{subfigure}[t]{0.24\textwidth}
    \centering
    \includegraphics[width=\linewidth]{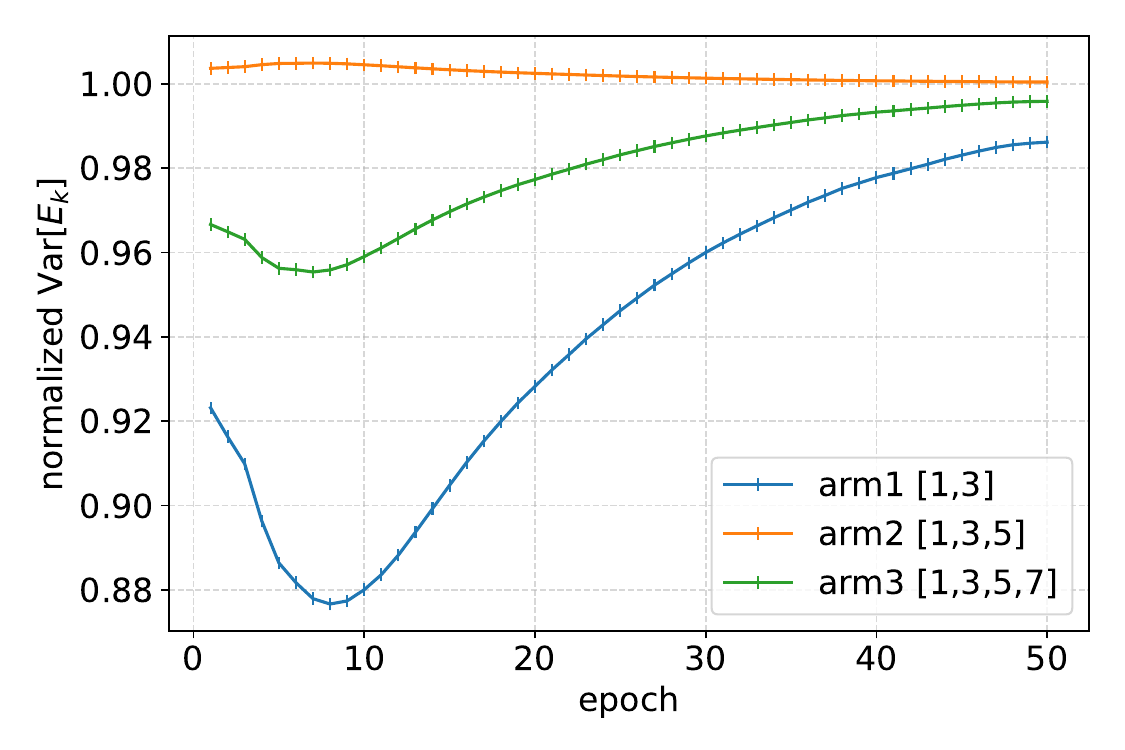}
    \caption{\footnotesize CIFAR-10: variance term $\mathrm{Var}[E_k]$ per arm}
    \label{fig:cifar-var2}
  \end{subfigure}

  \vspace{0.6em}

  \begin{subfigure}[t]{0.24\textwidth}
    \centering
    \includegraphics[width=\linewidth]{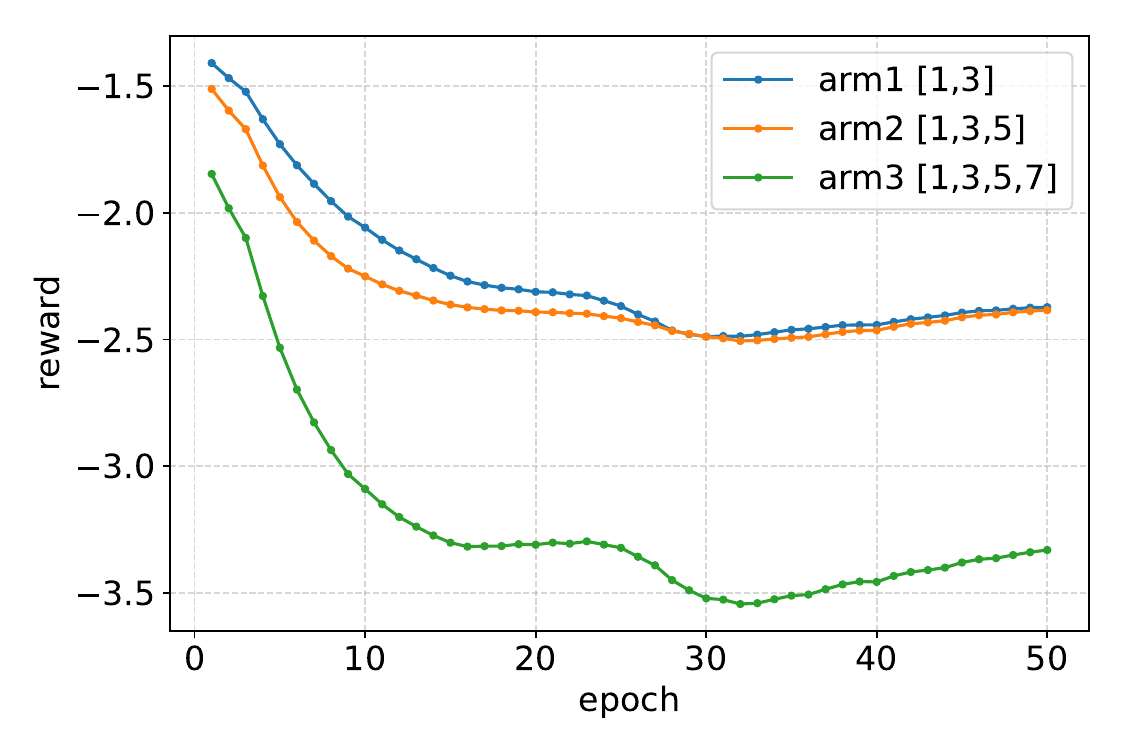}
    \caption{\footnotesize EuroSAT: reward per arm}
    \label{fig:euro-reward2}
  \end{subfigure}\hfill
  \begin{subfigure}[t]{0.24\textwidth}
    \centering
    \includegraphics[width=\linewidth]{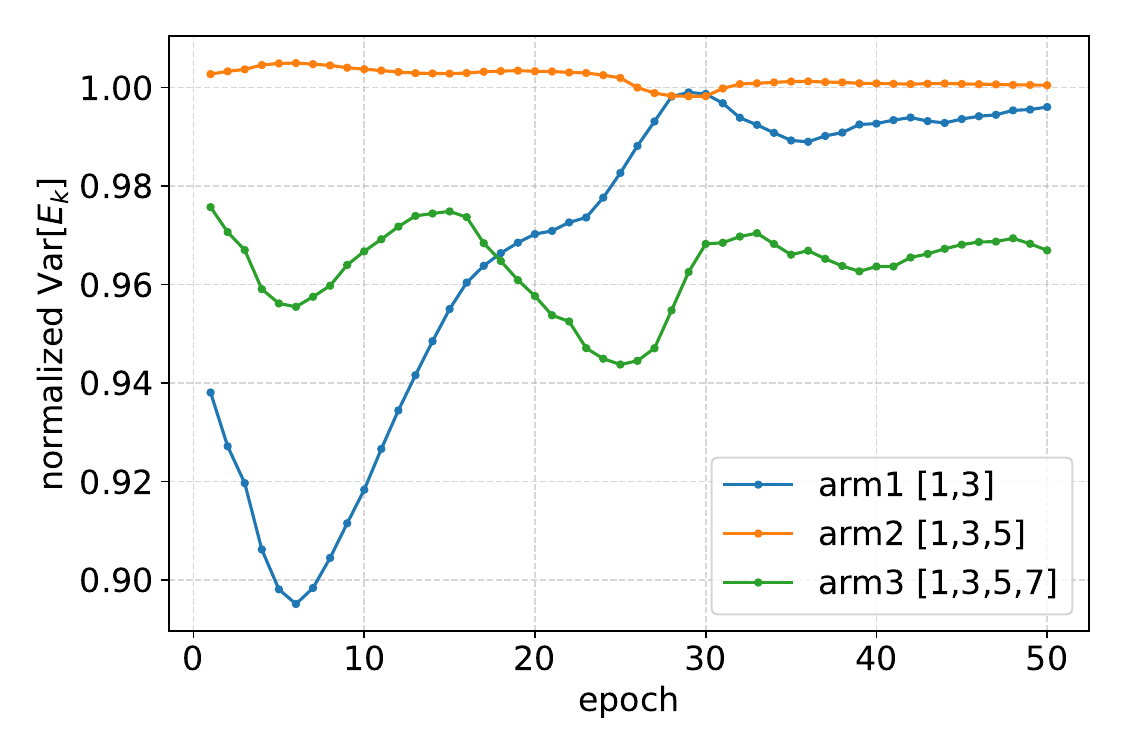}
    \caption{\footnotesize EuroSAT: variance term $\mathrm{Var}[E_k]$ per arm}
    \label{fig:euro-var2}
  \end{subfigure}

  \caption{\footnotesize Reward and variance across ZNE arm sets
  $\{1,3\}$, $\{1,3,5\}$, and $\{1,3,5,7\}$ for CIFAR-10 (top) and EuroSAT (bottom).
  Reward combines low-variance extrapolated estimates with a depth penalty, while the
  variance plots show measurement variability across epochs.}
  \label{fig:reward-variance-2x2}
\end{figure}

Fig.~\ref{fig:reward-variance-2x2} compares the safety-aware reward and normalized variance \(\mathrm{Var}[E_k]\) for ZNE arm sets \(\{1,3\}\), \(\{1,3,5\}\), and \(\{1,3,5,7\}\) on CIFAR-10 and EuroSAT. The reward balances estimator stability against the depth penalty. Across both datasets, the shallow arm \(\{1,3\}\) usually gives the highest reward because larger arms add execution overhead while providing only marginal variance reduction. This supports the CMAB design: deeper arms should be selected only when their extrapolation benefit justifies the additional cost.

\textbf{Communication \& Fidelity Summary} Table~\ref{tab:cmab-comm-fidelity} illustrates the communication and fidelity advantages achieved by CMAB-ZNE compared to the strongest fixed-ZNE baseline ([1,3,5]) under dynamic noise.  On CIFAR-10 and EuroSAT, the adaptive strategy has consistently selected parsimonious arm sets in the early and mid-training periods, decreasing average QPU-CPU round trips by around 40\% and exchanged bytes by 35\%, resulting in about 30\% lower end-to-end cost under a 10 Mbps budget.  In parallel, estimator fidelity improved by 6-7\% relative to the noiseless baseline, proving that adaptivity not only reduces communication overhead but also improves the accuracy of zero-noise extrapolation.  These results confirm that the CMAB formulation provides a better accuracy-efficiency trade-off over static ZNE approaches.

\begin{table}[t]
\centering
\scriptsize
\color{black}
\caption{\footnotesize Communication and fidelity gains of CMAB-ZNE versus the strongest fixed-ZNE baseline under dynamic noise ($\hat{\eta}_t$). Improvements are averaged across epochs and reported as mean $\pm$ standard deviation over 3 seeds at circuit depth 3 and noise band $\eta=0.05$ with a 10~Mbps link budget.}
\label{tab:cmab-comm-fidelity}
\setlength{\tabcolsep}{1.5pt}
\renewcommand{\arraystretch}{1.12}
\begin{tabular}{lccccc}
\toprule
Dataset & Baseline arm & \(\downarrow\) Round trips & \(\downarrow\) Bytes & \(\downarrow\) E2E cost & \(\uparrow\) Est. fidelity \\
 & (\([1,3,5]\)) & (\,\%) & (\,\%) & (\,\%) & (\,\%) \\
\midrule
CIFAR-10 & [1,3,5] & \(38.7 \pm 1.6\) & \(34.2 \pm 1.4\) & \(29.5 \pm 1.2\) & \(6.9 \pm 0.4\) \\
EuroSAT  & [1,3,5] & \(41.2 \pm 1.8\) & \(35.7 \pm 1.5\) & \(30.1 \pm 1.3\) & \(6.4 \pm 0.5\) \\
\bottomrule
\end{tabular}
\end{table}

{\color{black}
\textbf{Comparison with State-of-the-art Approaches.}
Table~\ref{tab:final-cmab-zne} evaluates all competing methods under the same simulation protocol to ensure a fair and repeatable comparison, instead of directly replicating values from the original publications. Specifically, all baselines use the same datasets, train/test splits, optimizer, shot budget, noise channels, classical encoder, variational quantum circuit (VQC) architecture, and dynamic noise schedule $\widehat{\eta}_t$. The compared methods include Clifford data regression (CDR)~\cite{strikis2021learning}, virtual distillation (VD)~\cite{huggins2021virtual}, probabilistic error cancellation (PEC)~\cite{cai2023quantum}, randomized compiling (RC)~\cite{jain2023improved}, reinforcement-learning-based mitigation~\cite{bordoni2024quantum}, hardware-aware optimization~\cite{niu2020hardware}, standard ZNE~\cite{temme2017error}, noise-aware folding~\cite{hour2024improving}, global-folding ZNE~\cite{sohn2025application}, and Adaptive KIK~\cite{koenig2025adaptive}. For CDR, we learn a regression map from noisy to low-noise observables using near-Clifford surrogate circuits. For VD, we use repeated noisy-state evaluations to approximate dominant-state purification. For PEC, we apply an inverse-noise correction inspired by quasi-probability using the same calibrated noise channels. For RC, we use randomized Pauli twirling and average the resulting observables. The ZNE-family baselines differ in how the folding set is selected, but they use the same folding operations and scale candidates as CMAB--ZNE. VD and PEC remain sensitive to the shot budget, backend noise model, quasi-probability normalization, and state-copy assumptions.
}

Table~\ref{tab:final-cmab-zne} shows that CMAB--ZNE achieves the best accuracy--loss trade-off on both datasets, reaching \(85.58\%\)/\(0.423\) on CIFAR-10 and \(96.02\%\)/\(0.157\) on EuroSAT. The strongest baselines are Adaptive KIK, hardware-aware optimization, and noise-aware ZNE variants, confirming that adaptive or hardware-informed mitigation is more effective than fixed extrapolation under dynamic noise. Standard ZNE, CDR, PEC, RC, and generic RL-based mitigation are less competitive in this setting, suggesting that mitigation strategies must explicitly account for time-varying noise and execution overhead. Overall, CMAB--ZNE performs best because it adapts the folding-scale set to the current noise context while penalizing unnecessary circuit-depth and communication cost.

\begin{table}
\centering
\caption{\footnotesize {\color{black}Comparison across mitigation approaches under \emph{dynamic} noise ($\hat{\eta}_t$). All competing methods were implemented by us under the same simulation environment, datasets, VQC architecture, shot budget, optimizer, and noise schedule. We report test accuracy (\%) and loss on CIFAR-10 and EuroSAT. Higher accuracy and lower loss are better.}}
\label{tab:final-cmab-zne}
\setlength{\tabcolsep}{0.5pt}
\renewcommand{\arraystretch}{1.15}
\begin{tabular}{lcccc}
\toprule
\multirow{2}{*}{Approach} & \multicolumn{2}{c}{CIFAR-10} & \multicolumn{2}{c}{EuroSAT}\\
\cmidrule(lr){2-3}\cmidrule(lr){4-5}
& Acc (\%) & Loss & Acc (\%) & Loss \\
\midrule
Clifford Data Regression (CDR)~\cite{strikis2021learning}                 & 82.93 & 0.512 & 94.86 & 0.186 \\
Virtual Distillation (VD)~\cite{huggins2021virtual}                      & 84.21 & 0.468 & 95.27 & 0.176 \\
Probabilistic Error Cancellation (PEC)~\cite{cai2023quantum}        & 83.74 & 0.489 & 95.03 & 0.181 \\
Randomized Compiling (RC)~\cite{jain2023improved}                    & 81.67 & 0.542 & 93.79 & 0.205 \\
Reinforcement learning~\cite{bordoni2024quantum}                    & 80.58 & 0.571 & 92.64 & 0.228 \\
Hardware-aware cost optimization~\cite{niu2020hardware}                  & 85.07 & 0.441 & 95.61 & 0.168 \\
\midrule
ZNE (standard extrapolation)~\cite{temme2017error}                    & 83.95 & 0.476 & 94.88 & 0.183 \\
ZNE (best practices + Pauli whirling)~\cite{majumdar2023best}  & 84.72 & 0.458 & 95.34 & 0.175 \\
ZNE (noise-aware folding)~\cite{hour2024improving}             & 85.03 & 0.449 & 95.48 & 0.171 \\
ZNE (global folding + inversion)~\cite{sohn2025application}    & 84.37 & 0.467 & 95.11 & 0.178 \\
Adaptive KIK (adaptive folding+filtering)~\cite{koenig2025adaptive} & 85.24 & 0.444 & 95.69 & 0.167 \\
\midrule
\textbf{Ours (adaptive ZNE + dynamic noise)}   & \textbf{85.58} & \textbf{0.423} & \textbf{96.02} & \textbf{0.157} \\
\bottomrule
\end{tabular}
\end{table}

{\color{black}
\subsection{Real-Hardware Validation on Rigetti Cepheus-1-108Q}
\label{sec:hardware_validation}

We conducted an additional real-hardware experiment using Amazon Braket and the Rigetti Cepheus-1-108Q superconducting quantum processor to further validate the practical applicability of the proposed \textit{CMAB--ZNE} framework beyond simulator-based noise models. Instead of performing full QPU-based training, this experiment was designed as a cost-controlled hardware inference validation. Specifically, the reduced quantum inputs, trained quantum circuit parameters, labels, and classifier weights were frozen and exported after the hybrid model was first trained in simulation. ZNE reconstruction and classification were performed as classical postprocessing, while the Rigetti QPU was used only to execute the quantum inference circuit.
The hardware validation used 30 CIFAR-10 samples, two separate hardware episodes, four ZNE scale values \(\{1,3,5,7\}\), and 500 shots per scale. Consequently, the total number of submitted QPU tasks was
\[
30 \times 2 \times 4 = 240.
\]
All 240 tasks were completed successfully, with no failed executions. To ensure reproducibility, we recorded the Amazon Braket task ARN, device ARN, ZNE scale, shot count, measured bitstring counts, Pauli-\(Z\) expectation values, runtime, estimated cost, and arm-wise ZNE reconstruction results for every hardware execution.
For each completed task, the measured bitstring counts were converted into qubit-wise Pauli-\(Z\) expectation values. For qubit \(j\) and scale \(\lambda\), the expectation value was estimated as
\[
\hat{z}_{j,\lambda}
=
\frac{N_{0,j,\lambda}-N_{1,j,\lambda}}
     {N_{0,j,\lambda}+N_{1,j,\lambda}},
\]
where \(N_{0,j,\lambda}\) and \(N_{1,j,\lambda}\) denote the number of shots in which qubit \(j\) was measured as 0 and 1, respectively. Three ZNE arms were then reconstructed using the measured expectation vectors from the shared scale set \(\{1,3,5,7\}\):
\[
\Lambda_1=\{1,3\}, \qquad
\Lambda_2=\{1,3,5\}, \qquad
\Lambda_3=\{1,3,5,7\}.
\]
This shared-scale protocol avoids running each arm separately. Without scale reuse, evaluating all three arms would require \(30 \times 2 \times (2+3+4)=540\) QPU tasks. In contrast, the shared-scale protocol required only 240 tasks, reducing the number of hardware tasks by approximately \(55.56\%\).
The real-hardware results are reported in Table~\ref{tab:rigetti_hardware_results}. All three ZNE arms achieved the same mean accuracy of \(93.33\%\). The shallow arm \(\{1,3\}\) achieved the best cost-aware reward, while the intermediate arm \(\{1,3,5\}\) achieved the lowest loss. This result shows that using more ZNE scale points does not necessarily improve real-hardware accuracy. Instead, the shallow arm provided the best performance--cost trade-off while preserving the same accuracy. This reinforces the main motivation behind \textit{CMAB--ZNE}: adaptive ZNE should select arms based on both mitigation quality and execution overhead, rather than always choosing the deepest or most expensive folding configuration.
\begin{table}[t]
\centering
\caption{\textcolor{black}{Real-hardware validation results on Rigetti Cepheus-1-108Q using CIFAR-10 samples. The experiment used 30 samples, two hardware episodes, 500 shots per scale, and the shared scale set \(\{1,3,5,7\}\).}}
\label{tab:rigetti_hardware_results}
\setlength{\tabcolsep}{3pt}
\renewcommand{\arraystretch}{1.05}
\begin{tabular}{lccccc}
\toprule
\textbf{ZNE Arm} & \textbf{Acc.} & \textbf{Loss} & \textbf{Reward} & \textbf{Entropy} & \textbf{Cost} \\
\midrule
\(\{1,3\}\)       & 93.33\% & 0.3304 & -0.3660 & 0.1133 & \$61.50 \\
\(\{1,3,5\}\)     & 93.33\% & 0.3295 & -0.3851 & 0.1137 & \$92.25 \\
\(\{1,3,5,7\}\)   & 93.33\% & 0.3300 & -0.4057 & 0.1138 & \$123.00 \\
\bottomrule
\end{tabular}
\vspace{-0.6em}
\end{table}
The hardware run took approximately 1210.01 seconds and incurred an estimated cost of \$369.00. The table also shows that deeper ZNE arms have higher estimated execution costs because they require more scale evaluations. Nevertheless, the shallow arm \(\{1,3\}\) matched the accuracy of the largest arm \(\{1,3,5,7\}\) while achieving the best cost-aware reward. This result demonstrates the need for cost-aware arm selection under realistic QPU cost constraints and validates the successful execution of the trained quantum inference circuit, ZNE scale evaluation, arm reconstruction, expectation-value estimation, and reward computation on a real superconducting QPU.
}

\subsection{Limitations}
{\color{black}
Although we added real-hardware validation on Rigetti Cepheus-1-108Q, the experiment remains cost-controlled and is limited to frozen quantum inference and ZNE arm reconstruction, rather than full QPU-based training. Broader validation is still needed across additional samples, calibration windows, shot budgets, QPU architectures, and hardware episodes. The simulated Lindblad-style model only partially captures the drift, non-Markovian effects, crosstalk, readout errors, and calibration artifacts that real NISQ devices may exhibit. Future work should incorporate on-device calibration, readout mitigation, backend-specific \(\lambda\)-scaling, telemetry signals such as \(T_1/T_2\), gate errors, and readout error rates. Finally, direct integration of hardware telemetry and calibration metadata may provide stronger CMAB context signals, while linear or quadratic Richardson fits may underfit highly nonlinear noise responses.
}

\section{Conclusions}
\label{sec:conclusion}

In this paper, we presented CMAB--ZNE, an adaptive noise-mitigation framework that improves the robustness of variational quantum circuits on NISQ hardware. Our method formulates ZNE folding-depth selection as a contextual bandit problem, where a per-epoch reward balances estimator variance and circuit-depth overhead using a dynamic noise proxy. This reformulates ZNE from a static post-processing step into an online policy that adapts to non-stationary noise.
Simulations and experiments on \textsc{EuroSAT} and \textsc{CIFAR-10} demonstrate that CMAB--ZNE consistently narrows the gap to the noiseless target more efficiently than fixed-$\lambda$ baselines under mixed-state simulations with Lindblad-modeled amplitude damping and dephasing. The adaptive policy automatically switches between shallow and deep arms according to the noise condition, resulting in a better accuracy--efficiency trade-off. CMAB--ZNE reduces quantum circuit execution round trips by up to 40\%, bytes transferred by 35\%, and end-to-end cost by 30\% over a 10 Mbps link. It also improves estimator fidelity by approximately 7\% compared with fixed ZNE. 


\ifCLASSOPTIONcaptionsoff
  \newpage
\fi

\bibliography{main}
\bibliographystyle{IEEEtran}
\begin{IEEEbiographynophoto}{Ratun Rahman} 
is a Ph.D. candidate in the Department of Electrical and Computer Engineering (ECE) at The University of Alabama in Huntsville, USA. His work focuses on machine learning, federated learning, and quantum machine learning. He has published papers in several IEEE journals on TNNLS, TC, TVT, IoTJ, TNSE, TETCI, Network, CL, and GRSL, and conferences, including NeurIPS and CVPR workshops, IEEE QCE, and IEEE CCNC. 
\end{IEEEbiographynophoto}

\vspace{-2em}

\begin{IEEEbiographynophoto}{Dinh C. Nguyen}
is an assistant professor at the Department of Electrical and Computer Engineering, The University of Alabama in Huntsville, USA. He obtained the Ph.D. degree in computer science from Deakin University, Australia in 2021. His current research interests include quantum machine learning, Internet of Things, wireless networking with over 70 publications.
\end{IEEEbiographynophoto}
\end{document}